\documentclass{article}

\usepackage[main, final]{neurips_2026}

\usepackage[utf8]{inputenc} 
\usepackage[T1]{fontenc}    
\usepackage{hyperref}       
\hypersetup{colorlinks=true, linkcolor=red, citecolor=blue, urlcolor=magenta}
\usepackage{url}            
\usepackage{booktabs}       
\usepackage{amsfonts}       
\usepackage{nicefrac}       
\usepackage{microtype}      
\usepackage{xcolor}         

\usepackage{amsmath}
\usepackage{amssymb}
\usepackage{mathtools}
\usepackage{amsthm}
\usepackage{multirow}
\usepackage{enumitem}

\usepackage[capitalize,noabbrev]{cleveref}

\RequirePackage{algorithm}
\RequirePackage{algorithmic}

\theoremstyle{plain}
\newtheorem{theorem}{Theorem}[section]

\theoremstyle{definition}
\newtheorem{definition}[theorem]{Definition}

\theoremstyle{remark}

\usepackage{xcolor}

\newcommand{\stdsub}[1]{_{\scriptsize\textcolor{black!70}{\pm#1}}}

\newcommand{\best}[1]{%
  \ensuremath{\colorbox{orange!20}{\textbf{\textcolor{black}{#1}}}}%
}

\newcommand{\second}[1]{%
  \ensuremath{\colorbox{blue!15}{\textbf{\textcolor{black}{#1}}}}%
}

\usepackage{colortbl}
\definecolor{rowgray}{gray}{0.98}

\usepackage{xstring}

\newcommand{\AccMean}[1]{%
  \StrLeft{#1}{1}[\first]%
  \StrGobbleLeft{#1}{1}[\rest]%
  \textbf{\first}\rest%
}

\newcommand{\LatMean}[1]{%
  \StrLeft{#1}{1}[\first]%
  \StrGobbleLeft{#1}{1}[\rest]%
  \textbf{\first}\rest%
}

\title{Optimize Once, Execute Fast:\\ Latency-Aware Multi-Agent Workflow Learning for Recurrent Queries}

\author{%
  Enpei Zhang$^{1}$ \quad
  Feiyu Qu$^{1}$ \quad
  Zheng Huang$^{1}$ \quad
  Dawei Zhou$^{2}$ \quad
  Elynn Chen$^{3}$ \quad
  Yujun Yan$^{1}$ \\
  $^{1}$Dartmouth College \quad
  $^{2}$Virginia Tech \quad
  $^{3}$New York University \\
  \texttt{\{enpei.zhang.gr, feiyu.qu.gr, zheng.huang.gr, yujun.yan\}@dartmouth.edu} \\
  \texttt{zhoud@vt.edu} \quad
  \texttt{elynn.chen@stern.nyu.edu} \\
}

\newcommand\blfootnote[1]{%
  \begingroup
  \renewcommand\thefootnote{}\footnote{#1}%
  \addtocounter{footnote}{-1}%
  \endgroup
}

\begin{document}

\maketitle

\begin{abstract}
Large Language Model-based Multi-Agent Systems (LLM-based MAS) have shown remarkable success in solving complex tasks by coordinating specialized agents through multi-step workflows. To overcome the high cost and specificity of manually designed workflows, recent research has shifted toward learning them automatically. However, a key limitation remains: while agent calls introduce significant latency due to model inference time and service congestion, most automated methods optimize solely for downstream performance, ignoring execution time (makespan). This issue is amplified when workflows are reused across recurring queries, making latency accumulate over time. This limits deployment in recurrent applications such as AI customer support and AI cloud troubleshooting, where users expect fast responses. To bridge this gap, we introduce a new problem, termed latency-aware workflow generation, where the objective is to minimize workflow makespan while aiming to retain downstream performance. We propose \textbf{LAWA} (\textbf{L}atency-\textbf{A}ware \textbf{W}orkflow Optimiz\textbf{A}tion), a novel framework that reduces makespan via theoretically grounded graph edits targeting critical-path bottlenecks. By unifying a graph generator with this graph-editing procedure, LAWA supports both de novo workflow generation and refinement of existing ones. Empirical results across seven benchmarks demonstrate that LAWA significantly reduces makespan while achieving competitive task performance, supporting efficient reuse of multi-agent workflows for recurrent queries.
\blfootnote{Correspondence to \texttt{enpei.zhang.gr@dartmouth.edu}.}
\blfootnote{Our codes are available at \url{https://github.com/GraphmindDartmouth/LAWA}.}
\end{abstract}

\section{Introduction}

Large Language Model-based Multi-Agent Systems (LLM-based MAS) have shown strong capabilities in complex tasks such as code generation and reasoning by coordinating specialized agents through structured, multi-step workflows~\cite{autogpt,metagpt}. Meanwhile, manually designing such workflows requires substantial effort and domain expertise, motivating recent research on automated workflow generation~\cite{aflow,assemble_your_crew}, which allows MAS to automatically customize its workflow for different tasks. 

However, most prior workflow generation methods still optimize workflows primarily for downstream task performance, while largely overlooking execution time. This limitation is increasingly problematic in LLM-based MAS, where each agent call can incur considerable latency due to model inference and service congestion~\cite{kairos_llm_serving}, and the overall execution time (makespan) of the workflow can grow rapidly with increasing workflow depth~\cite{seqopt}. For workflows repeatedly invoked for recurrent queries, small per-query delays can accumulate into substantial latency costs. Such workloads arise in AI customer support, automated software engineering, and cloud troubleshooting~\cite{wang2025ecombench,metagpt,chen2025aiopslab}, where similar requests recur and users expect fast responses. Therefore, focusing solely on task performance may lead to workflows that are impractical for real-world deployment. To bridge this gap, we formally propose a new problem: latency-aware workflow generation, where the objective is to minimize makespan while preserving task performance.

To solve this problem, we propose \textbf{LAWA} (\textbf{L}atency-\textbf{A}ware \textbf{W}orkflow Optimiz\textbf{A}tion), a theoretically grounded framework for latency-aware workflow optimization, designed to reduce makespan with monotonicity guarantees while preserving downstream task performance. LAWA is built on three core components: First, we devise a workflow optimization engine that performs makespan-guided beam search over graph edits with theoretical guarantees of monotonic makespan reduction. Second, we train a Graph Evaluator that estimates downstream performance of candidate workflows and guides edit selection toward edits that maintain downstream task performance. Third, we train a Graph Generator that constructs initial chain workflows when no initial workflow is given, enabling LAWA to generate new workflows from scratch. Empirical results across seven benchmarks confirm that LAWA significantly reduces makespan while maintaining competitive task performance.

Our contributions are summarized as follows:
\begin{itemize}[leftmargin=*]
\item \textbf{New Problem:} We define latency-aware workflow generation, where makespan is treated as a primary optimization objective alongside task performance in time-sensitive multi-agent systems.
\item \textbf{New Theoretically Grounded Framework:} We propose LAWA, a framework that applies theoretically grounded graph editing operations, providing formal guarantees on makespan reduction.
\item \textbf{Empirical Validation:} We conduct comprehensive experiments across seven benchmarks, demonstrating that LAWA achieves significant makespan reduction compared to state-of-the-art baselines while maintaining competitive task accuracy.
\end{itemize}

\section{Problem Formulation}
\label{sec:problem}
In this section, we formalize the problem of latency-aware workflow optimization. We first introduce the notation for multi-agent workflows and the definition of workflow makespan, and then present our optimization objective.

\paragraph{Notations.}
We model a multi-agent workflow as a directed acyclic graph (DAG), $G = (\mathcal{V}, \mathbf{A}, \mathbf{t})$, where $\mathcal{V}$ is the node set, $\mathbf{A}$ is the adjacency matrix, and $\mathbf{t}$ is the vector of node latencies. Each node $v_k \in \mathcal{V}$ represents a single agent call with execution time $t_k > 0$. Different nodes may represent repeated calls to the same agent at different steps. The binary adjacency matrix $\mathbf{A}\in\{0,1\}^{|\mathcal{V}|\times|\mathcal{V}|}$ encodes directed information dependencies among agent calls: $\mathbf{A}_{ij}=1$ indicates that the output of node $v_i$ is provided as an input context to node $v_j$. We also introduce two special nodes, $v_I$ and $v_O$, denoting workflow entry and exit. The entry node $v_I$ has no incoming edges, the exit node $v_O$ has no outgoing edges, and every other node lies on at least one path from $v_I$ to $v_O$.

\paragraph{Makespan Definition.} Under full parallelism, each node starts immediately after all its predecessors have completed, and the makespan $T(G)$ (i.e., end-to-end runtime of the workflow) is determined by the critical path (i.e., the path from $v_I$ to $v_O$ with the maximum cumulative latency). Formally,
\begin{align}
T(G) = \max\limits_{P \in \mathcal{P}_G(v_I, v_O)} \sum_{v_k \in P} t_k,
\label{eq:def_T}
\end{align}
where $\mathcal{P}_G(v_I, v_O)$ denotes the set of all paths from the entry node $v_I$ to the exit node $v_O$ in $G$. 
In the main text, we only consider the above definition; an additional experiment with a limited number of executors is reported in Appendix~\ref{appsec:finite_resource}.

\paragraph{Optimization Goal.}
We define the optimization objective for latency-aware workflows as:
\begin{align}
\min_G~T(G), \quad
\text{s.t.}~~S(G) \ge S(G_0),
\label{eq:overall_goal}
\end{align}
where $S(G)$ denotes the downstream task performance of workflow $G$, and $G_0$ is the initial graph before modifications. The objective in Equation~\eqref{eq:overall_goal} minimizes the makespan $T(G)$, while the constraint enforces that the performance of the optimized workflow is no worse than that of the initial graph. In practice, the initial graph can be an existing workflow that meets performance requirements but is latency-inefficient, or a generated initial workflow. 

\begin{figure*}
    \centering
    \includegraphics[width=\linewidth]{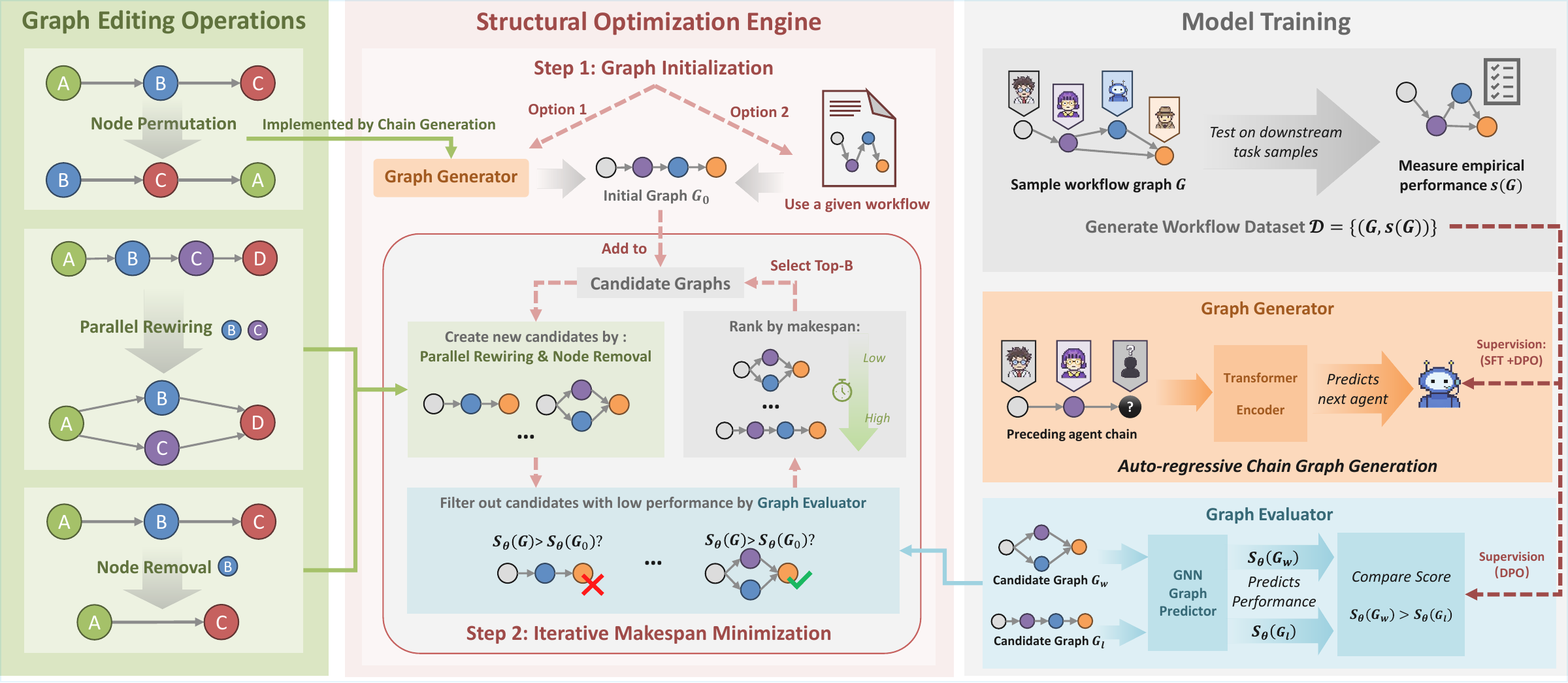}
    \caption{\textbf{Overview of LAWA.} Three main components of LAWA: (left) graph editing operations including node permutation, parallel rewiring, and node removal; (middle) a structural optimization engine that initializes a workflow, generates and ranks candidate graphs by makespan, and iteratively selects valid modifications; and (right) the training pipelines for the graph generator and graph evaluator based on empirical downstream task performance.}
    \label{fig:overview}
\end{figure*}

\section{Method}

\subsection{Overview}
In this section, we introduce \textbf{LAWA}, as shown in~\cref{fig:overview}, a latency-aware workflow generation framework that constructs and optimizes multi-agent workflows through iterative graph modifications. At its core, LAWA defines a set of \textbf{theoretically grounded graph modification operations} that characterize admissible structural changes for reducing makespan under workflow execution constraints. Building on this formulation, LAWA consists of a structural optimization engine and two learnable modules.
The structural optimization engine performs a makespan-guided search by iteratively applying valid graph modification operations to reduce workflow makespan.
When no initial graph is provided, LAWA includes a \textbf{Graph Generator}, which constructs initial workflows and enables LAWA to operate in a de novo generation setting.
To maintain downstream task performance during structural optimization, LAWA trains a \textbf{Graph Evaluator}, which estimates downstream task performance of candidate workflows and guides operation selection. In the following, we first introduce the graph editing operations and their theoretical properties, then present the detailed iterative structural optimization procedure, followed by the graph generator for workflow initialization and the graph evaluator used in the structural optimization. 
\subsection{Basic Graph Editing Operations}
\label{sec:ops}

We first define a set of graph editing operations that increase parallelism, remove redundant agents, and reorder agent execution, respectively.

\begin{definition}[Parallel Rewiring]
Given two nodes $i$ and $j$ with $\mathbf{A}_{ij}>0$, the \emph{parallel rewiring} operation removes the dependency from $i$ to $j$ and merges their outgoing and incoming edges to increase potential parallelism. Formally, we define the operator $\text{ParallelRewire}(G,i,j)$ as:
\begin{align}
\mathbf{A}_{ij} \leftarrow 0\;;\quad
\mathbf{A}_{ik} \leftarrow1\;\text{for all $k$ that $\mathbf{A}_{jk}=1$}\;;\quad
\mathbf{A}_{kj} \leftarrow1\;\text{for all $k$ that $\mathbf{A}_{ki}=1$}
\end{align}
\end{definition}
\begin{definition}[Node Removal]
The \emph{node removal} operation eliminates a non-terminal node $k\in\mathcal{V}\setminus\{v_I,v_O\}$ from the workflow while preserving connectivity by linking its predecessors and successors. Formally, we define the operator $\text{RemoveNode}(G,k)$ as:
\begin{align}
\mathbf{A}_{ij} &\leftarrow 1 \quad \text{for all $(i,j)$ that $\mathbf{A}_{ik}=1$ and $\mathbf{A}_{kj}=1$}\;;\\
\mathbf{A}_{ik} &\leftarrow 0 \quad \text{for all $i$ that $\mathbf{A}_{ik}=1$}\;;\\
\mathbf{A}_{kj} &\leftarrow 0 \quad \text{for all $j$ that $\mathbf{A}_{kj}=1$}\;.
\end{align}
\end{definition}
\begin{definition}[Node Permutation]
Let $\sigma:\mathcal{V}\to\mathcal{V}$ be a permutation over the node set.
We define the induced permutation operator $\Pi_\sigma$ acting on adjacency matrix $\mathbf{A}$ as
\begin{align}
\big(\Pi_\sigma(\mathbf{A})\big)_{\sigma(u),\,\sigma(v)} \;=\; \mathbf{A}_{u,v}
\quad \forall\, u,v \in \mathcal{V},
\end{align}
which reorders nodes in the workflow.
\end{definition}
The proposed graph operations are designed to satisfy three requirements. 
First, the operations preserve workflow connectivity. 
Second, the operations monotonically reduce workflow makespan. 
Finally, under a suitable node ordering, the operator set is expressive enough to transform a chain into any transitively reduced DAG considered in our formulation.
The following theorems establish these guarantees.
\begin{theorem}[Reachability]
\label{prop:reachability}
Let $u \leadsto_G v$ denote that there exists a directed path from $u$ to $v$ in $G$. Let $G'$ denote the graph obtained from $G$ after applying one of the operations:
\noindent\textnormal{(i) Node Removal.}
Let $k \in \mathcal{V}\setminus\{v_I,v_O\}$ be a removed node. Then,
\begin{equation}
\forall u, v \in \mathcal{V} \setminus \{k\}, \quad
u \leadsto_{G} v \;\iff\; u \leadsto_{G'} v .
\end{equation}
\noindent\textnormal{(ii) Parallel Rewiring.}
Let $(i,j)$ be the rewired edge. Then,
\begin{equation}
\forall u, v \in \mathcal{V} \setminus \{i,j\}, \quad
u \leadsto_{G} v \;\iff\; u \leadsto_{G'} v .
\end{equation}
\end{theorem}
Proof is presented in Appendix~\ref{appsec:prof_reachability}.

\begin{theorem}[Non-increasing Makespan]
Let $G=(\mathcal{V}, \mathbf{A}, \mathbf{t})$ be a workflow DAG. Let $G'$ denote the graph obtained from $G$ by applying either the Node Removal operation to a node $k \in \mathcal{V}\setminus\{v_I,v_O\}$ or the Parallel Rewiring operation to an existing edge $(i, j)$. Then $T(G') \le T(G).$
\label{prop:Non-increasingMakespan}
\end{theorem}
Proof is presented in Appendix~\ref{appsec:prof_Non-increasingMakespan}. 

\begin{theorem}[Completeness]
Let $G_{\text{chain}}$ be an arbitrary chain DAG over all nodes in $\mathcal{V}$ with adjacency matrix $\mathbf{A}_{\text{chain}}$.
For any target DAG $G^\ast=(\mathcal{V}, \mathbf{A}^\ast)$ that is transitively reduced, there exist a node permutation $\sigma$ and a finite operation sequence $\mathcal{O}$ consisting only of Parallel Rewiring and Node Removal such that
$
\mathcal{O}\!\left(\Pi_\sigma(\mathbf{A}_{\text{chain}})\right) \;=\; \mathbf{A}^\ast .
$
\label{prop:completeness}
\end{theorem}
Proof is presented in Appendix~\ref{appsec:prof_completeness}.

Together, \cref{prop:reachability,prop:Non-increasingMakespan} motivate Parallel Rewiring and Node Removal as graph-modification operators that preserve workflow connectivity and do not increase makespan. \Cref{prop:completeness} further indicates that the edit operators are expressive enough, in principle, to transform a chain into any transitively reduced DAG under a suitable node ordering, and the Graph Generator (\cref{sec:generator}) is designed to learn initial chain workflows with such suitable node orderings.

\subsection{Latency-Aware Structural Optimization}
\label{sec:proce}
\begin{algorithm}[t]
{\small
\caption{Latency-Aware Structural Optimization}
\label{alg:lawa_overview}
\begin{algorithmic}[1]
\renewcommand{\algorithmicrequire}{\textbf{Input:}}
\renewcommand{\algorithmicensure}{\textbf{Output:}}

\REQUIRE Optional initial graph $G_0$, beam size $B$, max iterations $K$.

\STATE Initialize pretrained Graph Generator $\mathcal{G}_\phi$ and pretrained Graph Evaluator $S_\theta$

\STATE Set $G_0 \leftarrow \mathcal{G}_\phi()$ if $G_0=\varnothing$
\STATE Initialize candidate set $\mathcal{C} \leftarrow \{G_0\}$
\FOR{$k = 1$ to $K$}
    \FORALL{$G \in \mathcal{C}$}
        \STATE Construct candidates by Node Removal: $\mathcal{C} \leftarrow \mathcal{C} \cup \{\operatorname{RemoveNode}(G,u)\mid u \in \mathcal{V}(G)\setminus\{v_I,v_O\}\}$
        \STATE Compute $\mathbf{R}^{\mathrm{rwr}}$ by \cref{eq:r_rewire}
        \FORALL{edge $(u,v)$ with $\mathbf{R}^{\mathrm{rwr}}_{u,v}\!\ge\!0$}
            \STATE Construct candidates by Parallel rewiring: $\mathcal{C} \leftarrow \mathcal{C} \cup \{\operatorname{ParallelRewire}(G,u,v)\}$
        \ENDFOR
    \ENDFOR

    \STATE Performance-Based filtering: $\mathcal{C} \leftarrow \{G' \in \mathcal{C} \mid S_\theta(G') \ge S_\theta(G_0)\}$
    \STATE Candidate update by runtime-first ranking: $\mathcal{C} \leftarrow \text{Top-}B_{T(\cdot)}(\mathcal{C})$
\ENDFOR
\STATE Select final solution:
$\displaystyle 
G^\ast \leftarrow \arg\min_{G \in \mathcal{C}} \big( T(G) \big)$
\ENSURE Optimized workflow $G^\ast$.
\end{algorithmic}}
\end{algorithm}

Based on the graph editing operations defined above, we now describe the procedure for latency-aware structural optimization in LAWA, which operates in two stages.

\paragraph{Workflow Initialization.}
LAWA initializes workflows in two ways: it either refines an existing workflow graph, or constructs an initial graph using the Graph Generator when no prior structure is available. The Graph Generator produces chain-structured graphs, motivated by the expressivity result in \cref{prop:completeness}.
These two modes allow LAWA to function both as a plug-in optimizer for existing workflows and as a de novo workflow generator. 

\paragraph{Iterative Structural Optimization.} In the second stage, LAWA refines workflow structure via a guided beam search to reduce makespan (\cref{alg:lawa_overview}). Before the search, we first estimate node latencies $t_k$ by calibration runs (e.g., each agent's P95 response time on a calibration set). At each beam-search iteration, LAWA follows a three-stage procedure. First, for each workflow $G \in \mathcal{C}$, LAWA generates new candidates by applying one step of Node Removal or Parallel Rewiring. Node Removal is applied to every non-terminal node in $G$ excluding $v_I$ and $v_O$. For Parallel Rewiring, to reduce the number of candidates, we apply it only to edges $(u,v)$ with $\mathbf{R}^{\mathrm{rwr}}_{u,v} \ge 0$, where $\mathbf{R}^{\mathrm{rwr}}_{u,v}$ is the approximate change in downstream task performance induced by rewiring $(u,v)$. \emph{Parallel Rewiring is disallowed for node pairs invoking the same agent, as such nodes cannot be executed in parallel in reality.} Second, LAWA evaluates each candidate workflow $G' \in \mathcal{C}$ using a Graph Evaluator $S_\theta(\cdot)$ and retains only those satisfying the baseline performance constraint $S_\theta(G') \ge S_\theta(G_0)$. Finally, LAWA keeps the Top-$B$ retained workflows with the smallest makespan $T(G)$ as the candidate set $\mathcal{C}$ for the next iteration. To avoid repeatedly executing candidate workflows, the Graph Evaluator is used to predict downstream task performance; $\mathbf{R}^{\mathrm{rwr}}_{u,v}$ is computed from its gradient information.

\subsection{Graph Generator for Initialization}
\label{sec:generator}

The Graph Generator provides an initial chain-structured workflow when no prior structure is available. For each agent node $u$, we first construct an embedding $\mathbf{e}_u$ by encoding the agent role and system prompt with a frozen Sentence-BERT encoder. We formulate workflow initialization as an auto-regressive sequence generation problem, where each token corresponds to the agent embedding $\mathbf{e}_u$, and the model predicts the next token (agent) $a_t$ according to current sequence $a_{<t}$, producing an ordered agent sequence that defines a chain workflow.

\paragraph{Graph Generator Architecture.}
The Graph Generator adopts a transformer encoder $\mathcal{R}_\phi$ for auto-regressive sequence generation. Given a partial agent sequence $a_{<t}$, the encoder produces a contextual hidden representation $\mathbf{h}_t=\mathcal{R}_\phi(a_{<t}) \in \mathbb{R}^d$. The model predicts the next agent $a_t$ by computing a similarity-based distribution between $\mathbf{h}_t$ and the agent embedding matrix $\mathbf{E}$:
\begin{align}
P_\phi(a_t \mid a_{<t}) = \text{Softmax}\!\left( \frac{\mathbf{h}_t \cdot \mathbf{E}^\top}{\sqrt{d}} \right),
\end{align}
During generation, the entry node is used as the initial token, and generation terminates when the exit node is produced.

\textbf{Graph Generator Training.}
We train $\mathcal{R}_\phi$ using a two-stage strategy consisting of supervised pre-training and preference alignment. We use a graph dataset $\mathcal{D}=\{(G,s(G))\}$, where $s(G)$ is the empirical performance of workflow $G$ evaluated on the held-out training set for the corresponding task. From $\mathcal{D}$, we retain highest-scoring graphs $\mathcal{D}_{\text{high}}$ for supervised learning. For each $G \in \mathcal{D}_{\text{high}}$, we sample sequences $y=(a_1, \dots, a_L)$ such that for any indices $1\le i < j\le L$, the corresponding node of $a_i$ in $G$ is not reachable from that of $a_j$, and minimize the negative log-likelihood:
\begin{align}
\mathcal{L}_{\text{SFT}}(\phi) = -\mathbb{E}_{y \sim \mathcal{D}_{\text{high}}} \sum_{t=1}^{L} \log P_\phi(a_t \mid a_{<t}).
\end{align}
To further align generation with workflow performance, we also apply DPO. Preference pairs $(y_w,y_l)$ are constructed from graph $G_w,G_l\in\mathcal{D}$ with $s(G_w)>s(G_l)$, and $\mathcal{R}_\phi$ is trained using the DPO loss:
\begin{align}
\mathcal{L}_{\text{DPO}}(\phi)& = -\mathbb{E}_{(y_w,y_l)}
\left[\log \sigma\!\left(
\beta \log \frac{P_\phi(y_w)}{P_{\text{ref}}(y_w)} -
\beta \log \frac{P_\phi(y_l)}{P_{\text{ref}}(y_l)}
\right)\right],
\end{align}  
where $P_{\text{ref}}$ is the frozen SFT model and $\beta$ controls deviation from the reference policy.

\subsection{Graph Evaluator for Performance Prediction}
\label{sec:eval}

The Graph Evaluator is a supervised model that maps a workflow graph $G$ to a scalar score $S_\theta(G)\in\mathbb{R}$ to estimate the downstream task performance. As only the relative performance of the workflow matters, we train it using preference ranking instead of regression.

\paragraph{Graph Evaluator Architecture.}

We implement the Graph Evaluator $S_\theta$ as a graph neural network (GNN) over the workflow DAG. Each node $u\in\mathcal{V}$ is initialized with the agent embedding $\mathbf{e}_u$ defined in \cref{sec:generator}. The node embeddings are propagated through GNN message-passing layers and then pooled to obtain a graph-level representation, which is mapped to a scalar score by an MLP head:
\begin{equation}
S_\theta(G) = \mathrm{MLP}\!\left(\mathrm{Pool}\big(\mathrm{GNN}(\mathbf{A},\{\mathbf{e}_u\})\big)\right).
\end{equation}

\paragraph{Graph Evaluator Training.}
We train the Graph Evaluator using a preference-based objective. Specifically, we reuse the graph dataset $\mathcal{D}$ defined in \cref{sec:generator}. From $\mathcal{D}$, we construct preference pairs $(G_w, G_l)$ such that $s(G_w) > s(G_l)$. We then train the Graph Evaluator $S_\theta$ using a pairwise preference objective inspired by Direct Preference Optimization (DPO)~\cite{dpo}, which encourages higher scores for preferred workflows:
\begin{align}
\mathcal{L}_{\text{pairwise}} = - \log \sigma\!\left( \beta(\Delta_\pi - \Delta_{\text{ref}}) \right),
\end{align}
where $\sigma(\cdot)$ is the sigmoid function; $\beta\in\mathbb{R}$ is a hyperparameter; $\Delta_\pi = S_\theta(G_w) - S_\theta(G_l)$ denotes the predicted score difference; and $\Delta_{\text{ref}}$ is provided by a reference model (e.g., current best model on the validation set).

\paragraph{Gradient-Based Approximation for Edge Editing.}
To filter out rewiring operations that are unlikely to improve downstream task performance, we approximate the performance change of Parallel Rewiring using first-order gradients of the Graph Evaluator with respect to the workflow adjacency matrix $\mathbf{A}$. 
Let $\mathbf{\Gamma} = \partial S_\theta / \partial \mathbf{A}$ denote the gradient of the Graph Evaluator with respect to $\mathbf{A}$. 
For each potential rewiring $(i,j)$, the induced score change is approximated by the matrix inner product $\langle \mathbf{\Gamma}, \Delta \mathbf{A} \rangle$,
where $\Delta \mathbf{A}$ represents the change of $\mathbf{A}$ after the rewiring. The score change of all edges can be approximated by:
\begin{align}
\mathbf{R}^{\mathrm{rwr}} = -\, \mathbf{H} \;+\; \tilde{\mathbf{H}}\, \mathbf{A}^\top \;+\; \mathbf{A}^\top \tilde{\mathbf{H}},
\label{eq:r_rewire}
\end{align}
where $\mathbf{H} = \mathbf{\Gamma} \odot \mathbf{A}$, $\tilde{\mathbf{H}} = \mathbf{\Gamma} \odot (1-\mathbf{A})$, $\odot$ denotes the Hadamard product, and $\mathbf{R}^{\mathrm{rwr}}_{i,j}$ represents the estimated score change of rewiring edge $(i,j)$. The derivations are provided in Appendix~\ref{app:rewire_derivation}.

We further provide a theoretical analysis of a perturbed Graph Evaluator in Appendix~\ref{appsec:flip_prob}, showing that LAWA can still preserve performance with high probability under mild evaluator decision errors.

\section{Experiments}
In this section, we conduct extensive experiments to evaluate LAWA through four research questions.
(\textbf{RQ1}) How does LAWA, as a workflow generation framework, compare to prior methods in generating workflows?
(\textbf{RQ2}) Can LAWA, when used as a plug-in module, reduce workflow makespan while preserving performance?
(\textbf{RQ3}) How robust is LAWA when Graph Evaluator decisions are inaccurate or perturbed?
(\textbf{RQ4}) How do the individual components of LAWA contribute to its latency and performance?
Additional experiments are provided in Appendix~\ref{appsec:add_ablations}.

\subsection{Experimental Setup}

\noindent\textbf{Benchmarks.}
We conduct evaluations across seven benchmarks spanning three domains.
The code generation domain includes MBPP and HumanEval~\cite{MBPP,humaneval}.
The mathematical reasoning domain includes GSM8K, MATH500, SVAMP, and AQuA~\cite{gsm8k,math500,SVAMP,aqua}.
The general knowledge domain is evaluated using MMLU~\cite{MMLU}.
For each benchmark, we reserve a small subset of samples for training the proposed method and baselines, and hold them out from evaluation (see Appendix~\ref{appsec:exp_details}).

\noindent\textbf{Baselines.}
We compare LAWA against state-of-the-art dynamic multi-agent frameworks, including \textbf{DyLAN}~\cite{dylan}, \textbf{G-Designer}~\cite{gdesigner}, \textbf{AgentDropout}~\cite{agent_dropout}, \textbf{AgentPrune}~\cite{agentprune}, and \textbf{BiRouter}~\cite{birouter}, as well as the accuracy-oriented workflow generator \textbf{AFlow}~\cite{aflow}, and the latency-aware framework \textbf{LAMaS}~\cite{lamas_latency}. In addition, we include three basic workflow topologies:
(1) \textbf{Random Chain}, which randomly permutes all available agents into a sequential structure;
(2) \textbf{Random Graph}, which constructs a directed acyclic graph (DAG) using all agents, where edges are sampled with probability $p=0.5$ while enforcing acyclicity;
and (3) \textbf{3-Round Discussion}, which simulates a round-table discussion where all agents execute in parallel for three rounds, with the full context of the previous round broadcast to all agents in the subsequent round.

\noindent \textbf{Agent Backbones.} For a fair comparison, all methods use the same pool of agents built on \texttt{Qwen3-30B-A3B-Instruct-2507}~\cite{qwen3technicalreport}. Each agent is configured with different roles and functionalities. Across task domains, each method can access 8 candidate agents for coding, 10 for general knowledge, and 11 for mathematical reasoning. Additional details are provided in Appendix~\ref{appsec:exp_details}.

\textbf{Wall-Clock Makespan Measurement.} We report real end-to-end wall-clock makespan per query under the same computational budget; detailed settings are provided in Appendix~\ref{appsec:exp_details}. For LAWA, the $t_k$ used during optimization is estimated as each agent's P95 response time on training-set queries.

\begin{table*}[t]
\centering
\caption{\small \textbf{Real wall-clock makespan (Seconds$\downarrow$) per query on 7 benchmarks.} All methods are evaluated using agents powered by \texttt{Qwen3-30B-A3B-Instruct-2507} with different roles. Each method is tested three times and results are mean$\stdsub{\text{std}}$. The $\best{\textbf{first}}$ and $\second{\textbf{second}}$ best entries in each column are highlighted. LAWA achieves the lowest makespan on six of seven benchmarks and the best average rank.}
\vspace{1mm}
\label{tab:time_results}
\resizebox{0.99\textwidth}{!}{%
\begin{tabular}{lccccccc|c}
\toprule
\rowcolor{rowgray}
\textbf{Method}
& \textbf{MBPP}
& \textbf{HumanEval}
& \textbf{AQuA}
& \textbf{GSM8K}
& \textbf{MATH500}
& \textbf{SVAMP}
& \textbf{MMLU}
& \textbf{Avg. Rank} \\
\midrule
Rand. Graph
& $16.92\stdsub{1.53}$
& $19.80\stdsub{1.48}$
& $24.16\stdsub{0.42}$
& $17.13\stdsub{1.24}$
& $26.33\stdsub{1.41}$
& $17.24\stdsub{1.12}$
& $18.47\stdsub{0.43}$
& $8.86$ \\

\rowcolor{rowgray}
Rand. Chain
& $22.67\stdsub{1.52}$
& $22.26\stdsub{0.94}$
& $31.34\stdsub{2.65}$
& $27.90\stdsub{0.08}$
& $36.86\stdsub{0.29}$
& $23.54\stdsub{0.25}$
& $24.06\stdsub{0.22}$
& $10.57$ \\

3-Round Disc.
& $17.52\stdsub{2.84}$
& $15.08\stdsub{0.71}$
& $17.04\stdsub{0.61}$
& $15.06\stdsub{0.10}$
& $21.04\stdsub{0.94}$
& $13.60\stdsub{0.76}$
& $13.93\stdsub{0.13}$
& $7.29$ \\

\midrule
\rowcolor{rowgray}
DyLAN
& $11.46\stdsub{1.34}$
& $\second{11.29}\stdsub{0.08}$
& $13.65\stdsub{1.10}$
& $\second{8.09}\stdsub{0.19}$
& $16.63\stdsub{0.23}$
& $9.38\stdsub{0.61}$
& $9.11\stdsub{0.26}$
& $3.14$ \\

G-Designer
& $25.27\stdsub{1.75}$
& $24.01\stdsub{1.28}$
& $15.56\stdsub{0.26}$
& $9.74\stdsub{0.55}$
& $21.34\stdsub{9.46}$
& $23.47\stdsub{4.08}$
& $25.32\stdsub{2.43}$
& $9.43$ \\

\rowcolor{rowgray}
AgentDropout
& $13.92\stdsub{5.56}$
& $14.24\stdsub{3.73}$
& $13.85\stdsub{0.42}$
& $9.57\stdsub{0.95}$
& $21.54\stdsub{7.58}$
& $16.72\stdsub{2.80}$
& $17.31\stdsub{4.21}$
& $6.14$ \\

AgentPrune
& $20.18\stdsub{2.54}$
& $18.28\stdsub{1.24}$
& $14.56\stdsub{2.15}$
& $9.20\stdsub{0.08}$
& $19.25\stdsub{2.12}$
& $18.75\stdsub{0.69}$
& $19.01\stdsub{1.44}$
& $7.29$ \\

\rowcolor{rowgray}
BiRouter
& $13.54\stdsub{0.85}$
& $17.40\stdsub{3.51}$
& $14.19\stdsub{2.15}$
& $9.95\stdsub{1.12}$
& $18.64\stdsub{1.20}$
& $8.93\stdsub{0.54}$
& $12.17\stdsub{1.67}$
& $5.57$ \\

AFlow
& $10.24\stdsub{1.79}$
& $14.49\stdsub{2.86}$
& $15.53\stdsub{1.27}$
& $9.51\stdsub{0.57}$
& $18.55\stdsub{1.03}$
& $8.63\stdsub{0.84}$
& $11.00\stdsub{1.29}$
& $4.43$ \\

\rowcolor{rowgray}
LAMaS
& $\best{8.82}\stdsub{0.97}$
& $11.87\stdsub{1.38}$
& $\second{9.92}\stdsub{2.37}$
& $8.78\stdsub{1.80}$
& $\second{11.85}\stdsub{1.16}$
& $\second{8.12}\stdsub{0.88}$
& $\second{8.39}\stdsub{0.24}$
& $\second{2.14}$ \\

\midrule
\textbf{LAWA}
& $\second{8.84}\stdsub{0.79}$
& $\best{11.26}\stdsub{1.03}$
& $\best{9.86}\stdsub{0.38}$
& $\best{7.30}\stdsub{0.10}$
& $\best{11.80}\stdsub{0.48}$
& $\best{6.18}\stdsub{0.56}$
& $\best{7.86}\stdsub{1.28}$
& $\best{1.14}$ \\

\bottomrule
\end{tabular}}
\vspace{-3mm}
\end{table*}

\begin{table*}[t]
\centering
\caption{\small \textbf{Comparison of task performance (Accuracy$\uparrow$) on 7 Benchmarks.} All methods are evaluated using agents powered by \texttt{Qwen3-30B-A3B-Instruct-2507} with different roles. Each method is tested three times and results are mean$_{\pm \text{std}}$. The $\best{\textbf{first}}$ and $\second{\textbf{second}}$ best entries in each column are bolded with colors. LAWA maintains comparable or superior accuracy to the baselines, achieving the best average rank.}
\vspace{1mm}
\label{tab:accuracy_results}
\resizebox{0.99\textwidth}{!}{%
\begin{tabular}{lccccccc|c}
\toprule
\rowcolor{rowgray}
\textbf{Method}
& \textbf{MBPP}
& \textbf{HumanEval}
& \textbf{AQuA}
& \textbf{GSM8K}
& \textbf{MATH500}
& \textbf{SVAMP}
& \textbf{MMLU}
& \textbf{Avg. Rank} \\
\midrule

Rand. Graph
& $63.94\stdsub{0.49}$
& $92.07\stdsub{2.99}$
& $80.58\stdsub{0.67}$
& $93.23\stdsub{0.20}$
& $66.53\stdsub{0.41}$
& $93.67\stdsub{0.62}$
& $65.33\stdsub{4.94}$
& $7.36$ \\

\rowcolor{rowgray}
Rand. Chain
& $64.07\stdsub{0.73}$
& $91.67\stdsub{0.76}$
& $80.31\stdsub{0.64}$
& $93.91\stdsub{1.15}$
& $66.47\stdsub{0.57}$
& $93.30\stdsub{0.37}$
& $64.11\stdsub{5.60}$
& $8.00$ \\

3-Round Disc.
& $64.33\stdsub{1.12}$& $92.89\stdsub{0.76}$
& $72.57\stdsub{0.98}$
& $94.77\stdsub{0.21}$
& $63.80\stdsub{0.28}$
& $\best{94.27}\stdsub{0.29}$
& $65.66\stdsub{0.29}$
& $6.64$ \\

\midrule
\rowcolor{rowgray}
DyLAN
& $64.72\stdsub{1.02}$
& $\second{93.29}\stdsub{1.32}$
& $80.71\stdsub{0.96}$
& $95.80\stdsub{0.23}$
& $65.87\stdsub{0.41}$
& $93.67\stdsub{0.24}$
& ${68.03}\stdsub{0.09}$
& $4.79$ \\

G-Designer
& $63.81\stdsub{1.10}$
& $89.84\stdsub{0.76}$
& $76.77\stdsub{1.79}$
& $94.16\stdsub{0.33}$
& $66.13\stdsub{1.09}$
& ${94.20}\stdsub{0.29}$
& $67.71\stdsub{0.17}$
& $7.71$ \\

\rowcolor{rowgray}
AgentDropout
& ${64.59}\stdsub{1.91}$
& $91.06\stdsub{2.24}$
& ${80.84}\stdsub{1.30}$
& $93.58\stdsub{0.83}$
& $66.47\stdsub{0.09}$
& $93.27\stdsub{0.65}$
& $65.97\stdsub{2.15}$
& $7.07$ \\

AgentPrune
& $64.46\stdsub{0.66}$
& $90.45\stdsub{2.24}$
& $79.13\stdsub{3.36}$
& $94.06\stdsub{0.43}$
& $\best{67.33}\stdsub{0.77}$
& ${94.20}\stdsub{0.57}$
& $72.12\stdsub{5.66}$
& $5.64$ \\

\rowcolor{rowgray}
BiRouter
& $63.81\stdsub{0.84}$
& $92.48\stdsub{0.76}$& $82.15\stdsub{1.94}$
& ${95.05}\stdsub{0.28}$
& ${66.73}\stdsub{0.38}$
& $93.90\stdsub{1.00}$
& $65.09\stdsub{0.64}$
& $5.93$ \\

AFlow
& $\second{65.24}\stdsub{0.22}$
& $\best{93.70}\stdsub{0.93}$
& $\second{83.20}\stdsub{0.23}$
& $\best{96.06}\stdsub{0.95}$
& $66.80\stdsub{0.69}$
& $92.13\stdsub{0.83}$
& $\second{76.04}\stdsub{1.48}$
& $\second{3.14}$ \\

\rowcolor{rowgray}
LAMaS
& $64.98\stdsub{0.39}$
& $91.67\stdsub{0.93}$
& $79.53\stdsub{5.05}$
& $92.12\stdsub{1.14}$
& $65.33\stdsub{1.67}$
& $92.70\stdsub{2.17}$
& $66.99\stdsub{1.71}$
& $7.93$ \\

\midrule
\textbf{LAWA}
& $\best{66.28}\stdsub{0.37}$
& $92.89\stdsub{1.04}$
& $\best{83.46}\stdsub{0.85}$
& $\second{96.03}\stdsub{0.20}$
& $\second{67.00}\stdsub{0.86}$
& $\second{94.23}\stdsub{0.21}$
& $\best{77.77}\stdsub{3.15}$
& $\best{1.79}$ \\
\bottomrule
\end{tabular}}
\vspace{-3mm}
\end{table*}

\subsection{RQ1: LAWA as a Workflow Generator}

We first evaluate the effectiveness of LAWA as a workflow generator by comparing it with prior methods in terms of both execution efficiency and downstream task performance. In this setting, LAWA uses the Graph Generator to initialize workflows. To mitigate randomness, we execute all methods across 3 independent runs and report average results.

\paragraph{Makespan Analysis.}
Table~\ref{tab:time_results} reports real wall-clock makespan of generated workflows across seven benchmarks. LAWA is the fastest method on six of seven benchmarks and is only 0.02s slower than LAMaS on MBPP, while outperforming LAMaS in task accuracy on all seven benchmarks (Table~\ref{tab:accuracy_results}).
LAWA is also efficient to obtain: since the search operates on graph structures without executing agents, its one-time search and evaluator-training cost is recovered within a few recurrent executions, and even with graph-dataset construction included, the full pipeline remains substantially faster than the on-policy AgentDropout (Appendix~\ref{appsec:rq1_overhead}).

To intuitively demonstrate why LAWA achieves lower makespan, \cref{fig:graph_compare} visualizes three example workflow topologies and corresponding Gantt charts produced by representative methods on the MBPP benchmark. DyLAN adopts a layered, neural-network-like execution pattern that enables high parallelism within each round but enforces strict inter-round synchronization. This causes each round to wait for the slowest agent in the previous layer, leading to higher overall makespan. Another baseline, AgentPrune, simplifies workflows by pruning nodes and edges but does not explicitly target critical-path bottlenecks, often leaving a long critical path that dominates end-to-end makespan. In contrast, LAWA produces workflows with shorter critical paths. We include a case study on how LAWA reduces the critical path in Appendix~\ref{appsec:case_study}.

\paragraph{Performance Analysis.}
Table~\ref{tab:accuracy_results} summarizes task performance across all benchmarks. Despite aggressively optimizing workflow topology for makespan reduction, LAWA consistently maintains comparable or superior accuracy to the state-of-the-art baselines, achieving the best average rank among all methods, including the strong accuracy-oriented AFlow. These results show that, as a workflow generation method, LAWA can produce workflow graphs that achieve both low makespan and strong downstream performance. We also report total token consumption in Appendix~\ref{appsec:rq1_token}, where LAWA still remains token-efficient, and results with the frontier model Kimi-K2.6 in Appendix~\ref{appsec:frontier}.

\begin{figure*}[t]
    \centering
    \includegraphics[width=\linewidth]{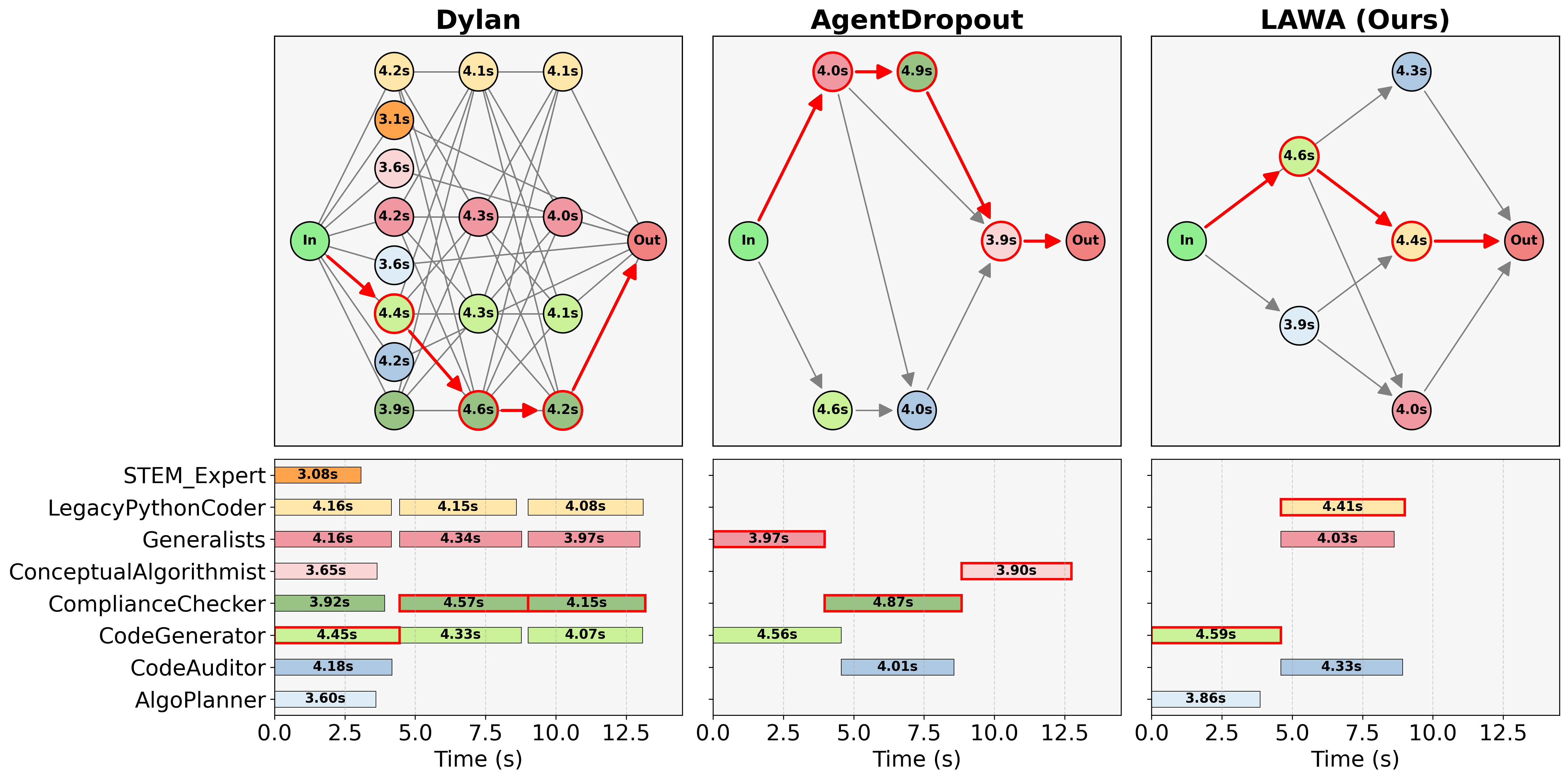}
    \vspace{-5mm}
    \caption{\small \textbf{Visualization of Workflow Topology and Execution Gantt Chart.} We visualize the workflow of three representative methods on one MBPP task. The top row shows the workflow graphs where red arrows denote the critical path, with the wall-clock latency labeled on each node. The bottom row shows the corresponding Gantt charts, with red-outlined blocks highlighting critical path tasks. LAWA generates the workflow with the shortest critical path to achieve maximal parallel execution and the shortest makespan.}
    \label{fig:graph_compare}
    \vspace{-3mm}
\end{figure*}

\begin{table*}[t]
\centering
\caption{\small \textbf{LAWA as a Plug-in Module.} We apply LAWA to the workflow graphs generated by four baseline methods and report the resulting changes in accuracy and makespan (mean$_{\pm \text{std}}$) relative to the original baseline performance. LAWA consistently compresses the makespan of these diverse initial workflows with non-negative mean accuracy changes in the reported runs, validating its role as a flexible plug-in module.}
\vspace{1mm}
\label{tab:ablation_transposed}
\setlength{\tabcolsep}{6.5pt}
\setlength{\cmidrulesep}{0.2pt}
\setlength{\aboverulesep}{0.7pt}
\setlength{\belowrulesep}{0.7pt}
\small
\resizebox{0.99\textwidth}{!}{%
\begin{tabular}{
l
>{\columncolor{rowgray}}c>{\columncolor{rowgray}}c
cc
>{\columncolor{rowgray}}c>{\columncolor{rowgray}}c
cc
}
\toprule
\multirow{2}{*}{\textbf{Dataset}}&
\multicolumn{2}{>{\columncolor{rowgray}}c}{\textbf{Random Graph}} &
\multicolumn{2}{>{\columncolor{white}}c}{\textbf{Random Chain}} &
\multicolumn{2}{>{\columncolor{rowgray}}c}{\textbf{AgentPrune}} &
\multicolumn{2}{>{\columncolor{white}}c}{\textbf{AgentDropout}} \\
\cmidrule(lr){2-3} \cmidrule(lr){4-5} \cmidrule(lr){6-7} \cmidrule(lr){8-9}
& {Acc (\%)} & {Makespan (s)}
& {Acc (\%)} & {Makespan (s)}
& {Acc (\%)} & {Makespan (s)}
& {Acc (\%)} & {Makespan (s)} \\
\midrule

HumanEval
& $\AccMean{+1.63}\stdsub{3.53}$ & $\LatMean{-4.40}\stdsub{0.97}$
& $\AccMean{+0.20}\stdsub{0.57}$ & $\LatMean{-7.61}\stdsub{0.04}$
& $\AccMean{+1.42}\stdsub{1.88}$ & $\LatMean{-2.79}\stdsub{0.46}$
& $\AccMean{+0.81}\stdsub{2.35}$ & $\LatMean{-2.19}\stdsub{1.56}$ \\

MBPP
& $\AccMean{+0.26}\stdsub{0.49}$ & $\LatMean{-5.39}\stdsub{1.49}$
& $\AccMean{+0.39}\stdsub{0.84}$ & $\LatMean{-6.86}\stdsub{0.06}$
& $\AccMean{+0.91}\stdsub{1.20}$ & $\LatMean{-2.41}\stdsub{0.72}$
& $\AccMean{+0.13}\stdsub{1.63}$ & $\LatMean{-2.75}\stdsub{0.73}$ \\

AQuA
& $\AccMean{+0.52}\stdsub{0.74}$ & $\LatMean{-4.92}\stdsub{1.59}$
& $\AccMean{+0.52}\stdsub{0.19}$ & $\LatMean{-8.05}\stdsub{0.11}$
& $\AccMean{+1.84}\stdsub{3.17}$ & $\LatMean{-3.65}\stdsub{0.06}$
& $\AccMean{+0.79}\stdsub{1.96}$ & $\LatMean{-3.21}\stdsub{0.12}$ \\

GSM8K
& $\AccMean{+0.25}\stdsub{0.25}$ & $\LatMean{-6.16}\stdsub{1.33}$
& $\AccMean{+0.08}\stdsub{1.07}$ & $\LatMean{-8.17}\stdsub{0.01}$
& $\AccMean{+0.08}\stdsub{0.75}$ & $\LatMean{-3.06}\stdsub{0.43}$
& $\AccMean{+1.26}\stdsub{1.26}$ & $\LatMean{-2.78}\stdsub{0.24}$ \\

MATH500
& $\AccMean{+0.47}\stdsub{1.57}$ & $\LatMean{-6.63}\stdsub{0.57}$
& $\AccMean{+0.47}\stdsub{0.57}$ & $\LatMean{-8.66}\stdsub{0.16}$
& $\AccMean{+0.07}\stdsub{0.90}$ & $\LatMean{-3.35}\stdsub{0.42}$
& $\AccMean{+0.13}\stdsub{0.52}$ & $\LatMean{-3.71}\stdsub{0.13}$ \\

SVAMP
& $\AccMean{+1.23}\stdsub{0.78}$ & $\LatMean{-5.53}\stdsub{1.43}$
& $\AccMean{+0.03}\stdsub{0.21}$ & $\LatMean{-8.17}\stdsub{0.01}$
& $\AccMean{+0.03}\stdsub{1.01}$ & $\LatMean{-3.70}\stdsub{0.33}$
& $\AccMean{+0.63}\stdsub{0.69}$ & $\LatMean{-2.92}\stdsub{0.73}$ \\

MMLU
& $\AccMean{+1.94}\stdsub{5.88}$ & $\LatMean{-4.61}\stdsub{0.33}$
& $\AccMean{+3.26}\stdsub{3.87}$ & $\LatMean{-8.62}\stdsub{0.29}$
& $\AccMean{+0.02}\stdsub{1.10}$ & $\LatMean{-3.14}\stdsub{0.13}$
& $\AccMean{+0.11}\stdsub{5.74}$ & $\LatMean{-3.48}\stdsub{0.58}$ \\
\bottomrule
\end{tabular}}
\vspace{-3mm}
\end{table*}

\begin{table*}[!t]
\centering
\caption{\small \textbf{RQ3: Robustness under evaluator perturbations.} We flip the evaluator's keep/reject decision with probability $\alpha$, and report the task accuracy (\%) of the generated workflows as mean$\stdsub{\text{std}}$. Performance remains close to the unperturbed setting for small flip rates ($\alpha \le 0.2$).}
\vspace{1mm}
\label{tab:rq3_flip}
{
\renewcommand{\arraystretch}{1.1}
\resizebox{0.5\textwidth}{!}{%
\begin{tabular}{lccc}
\toprule
\rowcolor{rowgray}
\textbf{Dataset}
& \textbf{$\alpha$=0.00}
& \textbf{$\alpha$=0.10}
& \textbf{$\alpha$=0.20} \\
\midrule
HumanEval
& $92.89\stdsub{1.04}$
& $92.48\stdsub{0.35}$
& $91.67\stdsub{1.53}$ \\

\rowcolor{rowgray}
MBPP
& $66.28\stdsub{0.37}$
& $66.28\stdsub{0.59}$
& $65.76\stdsub{0.67}$ \\

AQuA
& $83.46\stdsub{0.85}$
& $83.34\stdsub{0.82}$
& $82.83\stdsub{1.38}$ \\

\rowcolor{rowgray}
MATH500
& $67.00\stdsub{0.86}$
& $66.73\stdsub{0.31}$
& $65.87\stdsub{0.23}$ \\
\bottomrule
\end{tabular}}
}
\vspace{-3mm}
\end{table*}

\begin{table*}[!t]
\centering
\caption{\small \textbf{RQ4: Component ablations.} Changes in accuracy (\%) and makespan (s) relative to full LAWA. Removing any component reduces accuracy on all three benchmarks.}
\vspace{1mm}
\label{tab:component_ablation}
{
\renewcommand{\arraystretch}{1.1}
\resizebox{0.8\textwidth}{!}{%
\begin{tabular}{lcccccc}
\toprule
\rowcolor{rowgray}
& \multicolumn{2}{c}{\textbf{AQuA}}
& \multicolumn{2}{c}{\textbf{SVAMP}}
& \multicolumn{2}{c}{\textbf{MMLU}} \\
\cmidrule(lr){2-3} \cmidrule(lr){4-5} \cmidrule(lr){6-7}
\rowcolor{rowgray}
\textbf{Ablation}
& $\Delta$\textbf{Acc.} & $\Delta$\textbf{Time}
& $\Delta$\textbf{Acc.} & $\Delta$\textbf{Time}
& $\Delta$\textbf{Acc.} & $\Delta$\textbf{Time} \\
\midrule
w/o Node Removal & $-0.65$ & $+1.96$ & $-4.71$ & $+8.30$ & $-3.36$ & $+15.92$ \\
\rowcolor{rowgray}
w/o Parallel Rewiring & $-0.78$ & $+1.64$ & $-4.45$ & $+4.97$ & $-2.25$ & $+18.67$ \\
w/o SFT & $-2.75$ & $-1.96$ & $-0.68$ & $+2.37$ & $-1.66$ & $+5.51$ \\
\rowcolor{rowgray}
w/o DPO & $-3.41$ & $+7.26$ & $-3.40$ & $-0.36$ & $-6.81$ & $+13.08$ \\
\bottomrule
\end{tabular}}
}
\vspace{-3mm}
\end{table*}

\subsection{RQ2: LAWA as a Plug-in Module}

We further evaluate LAWA as a plug-in structural optimizer by applying it to workflows generated by different baseline methods, without modifying their original generation procedures. Specifically, we extract the workflow graphs produced by each baseline and use each graph as the initialization for LAWA. To mitigate randomness, each baseline is executed three times, and we report the average changes in performance and makespan before and after applying LAWA.

As shown in Table~\ref{tab:ablation_transposed}, LAWA consistently reduces the makespan of workflows produced by all baseline methods while showing no systematic degradation in mean task performance in the reported settings. This suggests that the proposed graph modification operators can compress makespan with limited impact on downstream accuracy. Notably, these gains are observed across diverse initial topologies, including both chain-structured and general DAG workflows. This indicates that LAWA’s optimization behavior does not rely on chain-structured initialization. Overall, the results demonstrate that LAWA can serve as a flexible plug-in module for latency-aware workflow optimization.

\subsection{RQ3: Robustness under Inaccurate or Perturbed Evaluator Decisions}

To evaluate robustness under an imperfect evaluator, we flip the evaluator's keep/reject decision with probability $\alpha$ at the filtering step in Algorithm~\ref{alg:lawa_overview} (line 12) for each candidate graph during optimization and report final task performance under small perturbation rates $\alpha \in \{0.0,0.1,0.2\}$. As shown in Table~\ref{tab:rq3_flip}, LAWA remains stable under mild evaluator perturbations, with performance close to the unperturbed setting for $\alpha \le 0.2$. These results suggest that occasional evaluator mistakes have limited impact on final performance under mild perturbations.

\subsection{RQ4: Component Ablations}

To isolate the sources of LAWA's latency and performance gains, we ablate the two graph-editing operations and the two training stages of the Graph Generator, and report the resulting changes relative to full LAWA in Table~\ref{tab:component_ablation}. Removing either Node Removal or Parallel Rewiring restricts the optimizer to a single type of edit; the resulting optimization terminates earlier and yields longer makespans. Both ablations also reduce accuracy on all three benchmarks, as the available structures are more limited. For Graph Generator training, removing either SFT or DPO reduces accuracy on all three benchmarks, with consistently larger degradation without DPO. Both training stages therefore contribute to performance, with DPO having the larger effect.

\section{Related Works}

\paragraph{Workflow Generation for LLM-MAS.}
Empowered by Large Language Models, LLM-based multi-agent systems leverage multiple specialized agents and structured coordination to solve complex tasks across domains~\cite{autogpt, li2023camelcommunicativeagentsmind, metagpt}. Early systems often rely on manually specified protocols, static interaction topologies, or SOP-style workflows~\cite{wu2023autogenenablingnextgenllm, qian2024chatdevcommunicativeagentssoftware, zhang2024proagentbuildingproactivecooperative}. To reduce manual design effort, recent work explores automated workflow generation and orchestration, including LLM-generated or modular workflow synthesis~\cite{aflow, robustflow, flow_modularized}, parallel or broadcast-style collaboration~\cite{dylan, seqopt}, routing-based agent selection~\cite{agentnet, birouter, masrouter}, and graph-based topology learning or pruning~\cite{agentprune, agent_dropout, gdesigner, assemble_your_crew}. Despite these advances, most methods primarily optimize task performance, communication efficiency, or token cost, rather than explicitly minimizing workflow makespan.

\paragraph{Makespan Optimization for LLM-MAS.}
Latency is an increasingly important bottleneck for LLM-MAS, since multi-step agent workflows introduce repeated LLM calls, synchronization barriers, and service-side congestion~\cite{kairos_llm_serving, seqopt}. Existing latency-oriented studies mainly address the LLM serving layer~\cite{kairos_llm_serving, halo_batch_query, aragog, chimera_serving, nalar, thunderagent}, while optimization of the workflow topology itself remains relatively underexplored. Recent topology learning methods such as LAMaS~\cite{lamas_latency} and AgentBalance~\cite{agentbalance} begin to incorporate latency into MAS topology design; however, LAMaS focuses on predefined layer-wise parallel orchestration, while AgentBalance mainly treats latency as a budget constraint without a theoretical analysis of topology-level makespan reduction. LAWA fills this gap by optimizing arbitrary DAG-structured LLM-MAS workflows and providing theoretical characterizations of graph edits in terms of reachability preservation, non-increasing makespan, and expressivity.

\section{Conclusion}

In this paper, we present LAWA, a latency-aware framework for optimizing Large Language Model-based multi-agent workflows. LAWA uses makespan-guided graph editing to shorten critical paths, while a learned evaluator guides the search toward workflows that preserve downstream task performance. Across seven benchmarks, LAWA produces lower-latency execution graphs and can also serve as a plug-in optimizer for existing workflows. We hope this work encourages further study of workflow-level latency optimization in multi-agent systems. Additional discussion of limitations is provided in Appendix~\ref{appsec:limitations}.


\newpage
\bibliographystyle{plain}
\bibliography{ref}  
\medskip


\newpage
\appendix
\section{Proofs of Theoretical Properties}
\label{appsec:proofs}

Here we provide formal proofs for the properties in the main text.

\subsection{Proof of \cref{prop:reachability}}
\label{appsec:prof_reachability}

\begin{proof}
We analyze the two operations separately:

\noindent\textbf{1. Node Removal (of $k$):}
Let $k\in\mathcal{V}\setminus\{v_I,v_O\}$ and $\mathcal{V}' = \mathcal{V} \setminus \{k\}$. We consider any $u, v \in \mathcal{V}'$.

\noindent$(\implies)$ \textit{Forward Direction:}
Assume $u \leadsto_{G} v$. Let $P$ be a path from $u$ to $v$ in $G$.
\begin{itemize}[leftmargin=*]
    \item If $k \notin P$, then $P$ consists entirely of nodes in $\mathcal{V}'$ and edges unchanged by the operation. Thus, $P$ exists in $G'$.
    \item If $k \in P$, the path contains a segment $x \to k \to y$. The removal operation explicitly adds the shortcut edge $x \to y$ to $G'$ whenever $x \to k$ and $k \to y$ exist in $G$. We can replace every occurrence of $x \to k \to y$ with the single edge $x \to y$, yielding a valid path in $G'$.
\end{itemize}

\noindent$(\impliedby)$ \textit{Reverse Direction:}
Assume $u \leadsto_{G'} v$. Let $P'$ be a path in $G'$.
Edges in $G'$ are either original edges from $G$ or newly created shortcut edges.
\begin{itemize}[leftmargin=*]
    \item If an edge $(x, y) \in P'$ exists in $G$, it represents the direct connection $x \to y$.
    \item If an edge $(x, y) \in P'$ is a new shortcut not in $G$, it implies that $x$ was a predecessor of $k$ and $y$ was a successor of $k$ in $G$. Thus, this edge corresponds to the path segment $x \to k \to y$ in $G$.
\end{itemize}

\noindent\textbf{2. Parallel Rewiring (of edge $i \to j$):}
We consider $u, v \in \mathcal{V} \setminus \{i, j\}$. Note that the operation removes $i \to j$ and updates connectivity such that $j$ inherits $i$'s predecessors (incoming merge) and $i$ inherits $j$'s successors (outgoing merge).

\noindent$(\implies)$ \textit{Forward Direction:}
Assume $u \leadsto_{G} v$. Let $P$ be a path in $G$.
\begin{itemize}[leftmargin=*]
    \item If $(i, j) \notin P$, the path from $u$ to $v$ remains intact in $G'$ since only edge $i \to j$ is removed.
    \item If $(i, j) \in P$, the path from $u$ to $v$ contains a segment $x \to i \to j \to y$.
    In $G'$, the operation sets $\mathbf{A}'_{xj} \leftarrow 1$ (since $x \to i$) and $\mathbf{A}'_{iy} \leftarrow 1$ (since $j \to y$).
    We can bypass the removed edge by using the new edge $x \to j$, creating a path $... \to x \to j \to y \to ...$ in $G'$.
\end{itemize}

\noindent$(\impliedby)$ \textit{Reverse Direction:}
Assume $u \leadsto_{G'} v$. Let $P'$ be a path in $G'$.
Any edge in $P'$ is either an original edge or a newly formed edge due to merging neighborhoods.
\begin{itemize}[leftmargin=*]
    \item A new edge of type $x \to j$ exists in $G'$ only if $x \to i$ existed in $G$. Thus, $x \to j$ in $P'$ corresponds to the segment $x \to i \to j$ in $G$.
    \item A new edge of type $i \to y$ exists in $G'$ only if $j \to y$ existed in $G$. Thus, $i \to y$ in $P'$ corresponds to the segment $i \to j \to y$ in $G$.
\end{itemize}
Since $u, v \neq i, j$, the path must start and end outside the rewired edge. By expanding any new edges back to their original segments involving $i \to j$, we recover a valid path in $G$.
\end{proof}

\subsection{Proof of \cref{prop:Non-increasingMakespan}}
\label{appsec:prof_Non-increasingMakespan}

\begin{proof}
Let $\mathcal{P}_G(v_I, v_O)$ denote the set of all directed execution paths from the entry node $v_I$ to the exit node $v_O$ in graph $G$. It suffices to show that for every path $P' \in \mathcal{P}_{G'}(v_I,v_O)$ in the edited graph, there exists a corresponding path $P \in \mathcal{P}_{G}(v_I,v_O)$ in the original graph such that
$$
\sum_{v \in P'} t_v \;\le\; \sum_{v \in P} t_v .
$$

\textbf{(i) Node Removal.}
Suppose a non-terminal node $k\in\mathcal{V}\setminus\{v_I,v_O\}$ is removed.
For any path $P' \in \mathcal{P}_{G'}(v_I, v_O)$:
\begin{itemize}[leftmargin=*]
    \item If $P'$ uses only original edges from $G$, then the same path exists in $G$, and the path duration is identical.
    \item If $P'$ uses a shortcut edge $(u,v)$ introduced by removing $k$, then $G$ contains the segment $u \to k \to v$. Replacing every such shortcut in $P'$ by its original two-edge segment yields a valid path $P$ in $G$, with
    $$
    \sum_{v \in P'} t_v
    \le \sum_{v \in P} t_v ,
    $$
    since $t_k > 0$.
\end{itemize}

\textbf{(ii) Parallel Rewiring.}
Suppose the edge $(i,j)$ is rewired.
For any path $P' \in \mathcal{P}_{G'}(v_I, v_O)$:
\begin{itemize}[leftmargin=*]
    \item If $P'$ uses only original edges from $G$, then the same path exists in $G$.
    \item Otherwise, every edge newly introduced by the rewiring is either of the form $x\to j$ where $x\to i$ existed in $G$, or of the form $i\to y$ where $j\to y$ existed in $G$. We can replace each new edge $x\to j$ by the segment $x\to i\to j$, and each new edge $i\to y$ by the segment $i\to j\to y$. This recovers a valid path $P$ in $G$ and only inserts nodes with positive latency. Hence
    $
    \sum_{v \in P'} t_v \le \sum_{v \in P} t_v .
    $
\end{itemize}

Thus, every execution path in $G'$ is no longer than a corresponding path in $G$. Taking the maximum over paths in $G'$ yields
\begin{align*}
T(G') 
= \max_{P' \in \mathcal{P}_{G'}(v_I, v_O)} \sum_{v \in P'} t_v
\le\max_{P \in \mathcal{P}_G(v_I, v_O)} \sum_{v \in P} t_v
= T(G).
\end{align*}
\end{proof}

\subsection{Proof of \cref{prop:completeness}}
\label{appsec:prof_completeness}

\begin{proof}
The transformation from an arbitrary chain to the target graph $G^\ast$ is achieved via a three-stage procedure. We prove the validity of each stage sequentially.

\vspace{1em}
\noindent\textbf{Stage 1: Topological Alignment (via Permutation).} \\
Let $\sigma:\mathcal{V}\to\mathcal{V}$ be a permutation such that applying the operator $\Pi_\sigma$ to the initial chain yields a sequence where, for every directed edge $(u,v)$ in $G^\ast$, node $u$ precedes $v$. This condition ensures that the linear ordering defined by the permuted chain is consistent with the \textbf{topological sort} of $G^\ast$~\cite{toposort}.

\vspace{1em}
\noindent\textbf{Stage 2: Node Alignment (via Removal).} \\
If the initial chain contains nodes not required by $G^\ast$, we apply \emph{Node Removal} to eliminate all extraneous nodes.
By reachability preservation, removing nodes does not alter the relative reachability among the remaining nodes.
After this stage, the graph is a chain defined exactly on the target node set $\mathcal{V}$, with an ordering consistent with $G^\ast$.
Therefore, the permuted chain after the node permutation is the topological sort of $G^\ast$.
Without loss of generality, we assign node indices according to this ordered chain, starting from $v_1$ to $v_n$.

\vspace{1em}
\noindent\textbf{Stage 3: Exact Edge Reconstruction (via Rewiring)} \\
With the nodes correctly aligned and ordered, we employ a \textit{Reverse Sweeping} strategy to recover the exact edge structure. We first define the construction, then prove its correctness.

\paragraph{I. The Construction.}
We execute the following iterative procedure to refine the graph structure:
\begin{description}
    \item[\textbf{For}] each node $v_i$ iterating from $v_{n-1}$ down to $v_1$ (reverse from end to start):
    \begin{itemize}
        \item Let $\mathcal{N}_{out}'(v_i)$ denote the current set of immediate successors of $v_i$.
        \item \textbf{While} there exists a neighbor $v_k \in \mathcal{N}_{out}'(v_i)$ such that $(v_i, v_k)$ not in $G^\ast$:
        \begin{itemize}
            \item \textbf{Do:} Apply \textit{Parallel Rewiring} to the edge $(v_i, v_k)$, which updates: 
            $$ \mathcal{N}_{out}'(v_i) \leftarrow (\mathcal{N}_{out}'(v_i) \setminus \{v_k\}) \cup \mathcal{N}_{out}'(v_k) $$
        \end{itemize}
        \item \textbf{End While} (All remaining immediate successors are now valid target neighbors).
    \end{itemize}
    \item[\textbf{End For}]
\end{description}

\paragraph{II. Correctness Verification.}
We prove that this process yields $G' \cong G^\ast$ using structural induction on the node index $i$ (descending from $n$ to $1$).

\begin{description}[leftmargin=*]
    \item[Inductive Hypothesis:] Assume that for all nodes $v_j$ in the downstream of $v_i$ ($j > i$), the outgoing edges are already identical to the target graph: $\mathcal{N}_{out}'(v_j) = \mathcal{N}_{out}^\ast(v_j)$. (Base case $v_n$ is trivial).

    \item[Inductive Step:] Consider node $v_i$. We verify the final state of $\mathcal{N}_{out}'(v_i)$ when the rewiring loop for $v_i$ terminates.
    \begin{enumerate}[leftmargin=*]
        \item \textbf{No False Positives:} The loop condition explicitly mandates the removal of any edge $(v_i, v_k)$ that does not exist in $G^\ast$. Thus, $\mathcal{N}_{out}'(v_i) \subseteq \mathcal{N}_{out}^\ast(v_i)$.
        
        \item \textbf{No False Negatives:} Consider any valid target neighbor $v_m \in \mathcal{N}_{out}^\ast(v_i)$. For the sake of contradiction, suppose that $v_m$ is \textbf{absent} from the current immediate successor set $\mathcal{N}_{out}'(v_i)$.
        \begin{itemize}[leftmargin=*]
            \item Since no Parallel Rewiring operation has been applied to $(v_i, v_m)$, the reachability is strictly preserved (\cref{prop:reachability}), the absence of a direct link implies that $v_m$ must be reachable via some intermediate neighbor $v_k \in \mathcal{N}_{out}'(v_i)$.
            \item Because the loop has terminated, $v_k$ is confirmed to be a valid target neighbor (i.e., $(v_i, v_k)$ exists in $G^\ast$).
            \item Furthermore, since $k > i$, the inductive hypothesis guarantees that the path segment $v_k \to \dots \to v_m$ exists in $G^\ast$.
            \item Consequently, $G^\ast$ contains both the direct edge $(v_i, v_m)$ and the indirect path $v_i \to v_k \to \dots \to v_m$. This structure violates the \textbf{transitive reduction} property of $G^\ast$.
            \item Thus, the assumption is false. No such blocking node $v_k$ can exist, and $v_m$ must be present in $\mathcal{N}_{out}'(v_i)$.
        \end{itemize}
    \end{enumerate}
\end{description}

Thus, $\mathcal{N}_{out}'(v_i) = \mathcal{N}_{out}^\ast(v_i)$. By induction, after the loop terminates at $v_1$, the graph $G'$ is exactly isomorphic to $G^\ast$.
\end{proof}

\section{Derivation of the Parallel Rewiring Score}
\label{app:rewire_derivation}

Let $\mathbf{A}\in\{0,1\}^{n\times n}$ denote the adjacency matrix of the workflow graph, and let
$\mathbf{\Gamma}=\frac{\partial S_\theta}{\partial \mathbf{A}}\in\mathbb{R}^{n\times n}$ be the gradient of the Graph Evaluator score $S_\theta$ with respect to $\mathbf{A}$.
For a graph edit that induces a perturbation $\Delta \mathbf{A}$, a first--order Taylor expansion yields
\begin{align}
\Delta S \;\approx\; \langle \mathbf{\Gamma}, \Delta \mathbf{A}\rangle
\;=\; \sum_{u,v} \mathbf{\Gamma}_{uv}\,\Delta \mathbf{A}_{uv}.
\label{eq:first_order}
\end{align}
Rewiring an edge $(i,j)$ removes the original edge $i\!\to\! j$ and merges the outgoing neighbors of $j$
into $i$, as well as the incoming neighbors of $i$ into $j$.
Accordingly, the perturbation $\Delta \mathbf{A}$ consists of three parts:
\begin{itemize}
    \item Removal of the edge $(i,j)$, giving $\Delta \mathbf{A}_{ij}=-\mathbf{A}_{ij}$;
    \item Addition of edges $(i,k)$ for all $k$ such that $\mathbf{A}_{jk}=1$ and $\mathbf{A}_{ik}=0$;
    \item Addition of edges $(k,j)$ for all $k$ such that $\mathbf{A}_{ki}=1$ and $\mathbf{A}_{kj}=0$.
\end{itemize}

Substituting these changes into \eqref{eq:first_order} gives
\begin{align}
\Delta S(i,j)
\;\approx\;
&-\mathbf{\Gamma}_{ij}\mathbf{A}_{ij} \notag\\
&+\sum_{k}\mathbf{\Gamma}_{ik}\,\mathbf{A}_{jk}(1-\mathbf{A}_{ik}) \notag\\
&+\sum_{k}\mathbf{\Gamma}_{kj}\,\mathbf{A}_{ki}(1-\mathbf{A}_{kj}),
\label{eq:rewire_expanded}
\end{align}
which is the stated first--order estimate of the score change induced by rewiring $(i,j)$.

For compactness, define
$\mathbf{H}=\mathbf{\Gamma}\odot\mathbf{A}$ and
$\tilde{\mathbf{H}}=\mathbf{\Gamma}\odot(1-\mathbf{A})$.
The three terms in \eqref{eq:rewire_expanded} can then be written in matrix form, yielding the
parallel--rewiring score matrix
\begin{align}
\mathbf{R}^{\mathrm{rwr}}
\;=\;
-\,\mathbf{H}
\;+\;
\tilde{\mathbf{H}}\,\mathbf{A}^{\top}
\;+\;
\mathbf{A}^{\top}\tilde{\mathbf{H}},
\end{align}
where each entry $\mathbf{R}^{\mathrm{rwr}}_{i,j}$ equals
$\Delta S(i,j)$ and provides the first--order estimated score change
for rewiring edge $(i,j)$.

\section{Details in Experiment Settings}
\label{appsec:exp_details}

\noindent\textbf{Agent Settings.}
We define specific agent pools for each domain to support specialized reasoning.

For the \textbf{Code Generation} domain (MBPP, HumanEval), the available agents include: \texttt{AlgoPlanner}, \texttt{CodeGenerator}, \texttt{CodeAuditor}, \texttt{ConceptualAlgorithmist}, \texttt{LegacyPythonCoder}, \texttt{ComplianceChecker}, \texttt{Generalists} and \texttt{STEM\_Expert}. For the \textbf{Mathematical Reasoning} domain (GSM8K, MATH500, SVAMP, AQuA), the candidate agents are: \texttt{MathReasoner}, \texttt{Calculator}, \texttt{STEM\_Expert}, \texttt{ProblemSimplifier}, \texttt{HeuristicMathSolver}, \texttt{Generalists}, \texttt{ConsensusBuilder}, \texttt{ExecutiveArbitrator}, \texttt{ContrarianThinker}, \texttt{Verifier}, and \texttt{LateralThinkingAnalyst}. For the \textbf{General Knowledge} domain (MMLU), the agent pool consists of: \texttt{LateralThinkingAnalyst}, \texttt{SocialScientist}, \texttt{Generalists}, \texttt{STEM\_Expert}, \texttt{MathReasoner}, \texttt{TheoreticalScientist}, \texttt{CriticalTheorist}, \texttt{HeuristicMathSolver}, \texttt{ContrarianThinker}, and \texttt{ConsensusBuilder}.

\noindent\textbf{Serving Environment.}
For wall-clock makespan measurements, each run uses a freshly launched and warmed-up SGLang server on an NVIDIA RTX PRO 6000 Blackwell Max-Q Workstation Edition GPU (96GB memory), with exclusive GPU allocation during measurement.

\noindent\textbf{Graph Dataset Construction.}
To train the Graph Evaluator and the Graph Generator for each task, we construct benchmark-specific graph datasets. For each benchmark, we generate a fixed pool of $N=2048$ random graphs using the agent candidates from the corresponding task domain.
To obtain the ground-truth performance labels (scores) for each graph $G$, we evaluate $G$ on a subset of task samples from that benchmark's training set. Specifically, we use $k=50$ samples for mathematical reasoning datasets (GSM8K, MATH500, SVAMP, AQuA) and $k=30$ samples for code generation (HumanEval, MBPP) and general knowledge (MMLU) datasets. Since HumanEval and MATH500 do not provide official training splits, we use equivalent training proxies: for HumanEval, we build HumanEval-compatible training samples from the MBPP train split; for MATH500, we sample from the MATH train split (levels 4--5) while excluding instances that overlap with the MATH500 test set. The average accuracy across these $k$ samples serves as the target label for training. We also use the same samples for the training-based baselines. 

\noindent\textbf{LAWA Parameters.}
For the optimization phases, we use fixed beam size of $B=32$.
The maximum number of optimization steps (iterations) is set to $K=11$ for the generation experiments and $K=4$ for the plugin-experiments due to reduced workflow depth.

Our implementation and scripts for reproducing the experiments are available at \url{https://github.com/GraphmindDartmouth/LAWA}.

\section{Case Study: Optimization Trajectory on MBPP.}
\label{appsec:case_study}

To illustrate the step-by-step optimization behavior of LAWA, we start from a chain-structured workflow on MBPP and apply LAWA with beam size $B=1$, so that it will only modify the single best candidate graph of the last iteration.
\Cref{fig:pruning_process} illustrates the optimization trajectory.
The original workflow achieves an accuracy of 63.9\% and exhibits a fully sequential dependency pattern, resulting in a long critical path.

In the first step, LAWA applies parallel rewiring between \texttt{Generalists} and \texttt{AlgoPlanner}, exposing latent parallelism without altering the agent set.
This change shortens the critical path and yields a slight accuracy improvement to 64.1\%.

In the second step, LAWA removes the second \texttt{ComplianceChecker} node.
This node contributes redundant computation along the critical path, and its removal further reduces the makespan while increasing accuracy to 65.6\%.

In the third step, LAWA performs parallel rewiring between \texttt{ComplianceChecker} and \texttt{CodeAuditor}. While this modification slightly reduces accuracy to 65.2\%, it eliminates a remaining sequential dependency and reduces the runtime contribution of \texttt{ComplianceChecker}.

Among these steps, LAWA incrementally compresses the critical path while maintaining stable task performance.

\begin{figure*}[ht]
    \centering
    \includegraphics[width=1\linewidth]{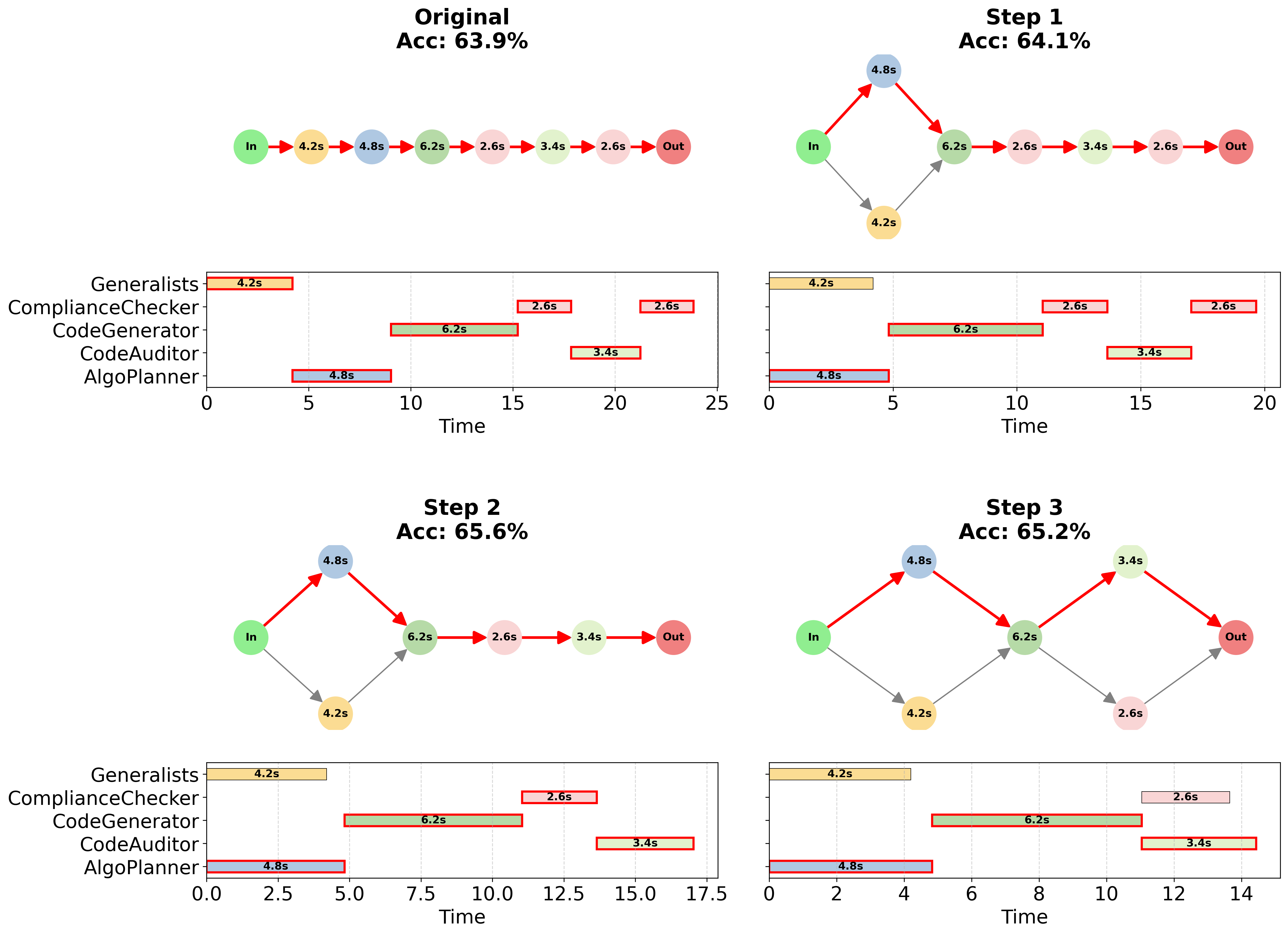}
    \caption{\textbf{Example of iterative workflow optimization on MBPP.}
    We visualize the original chain-structured workflow (top-left) and the resulting workflows after 1, 2, and 3 optimization steps (top-right, bottom-left, bottom-right, respectively). For each panel, the upper plot shows the workflow graph, where red arrows denote the critical path, and the lower plot shows the corresponding execution Gantt chart with critical-path agent executions highlighted (both using estimated $t_k$). As LAWA applies successive structural modifications, the workflow topology departs from the initial chain, progressively shortening the critical path and reducing makespan.}
    \label{fig:pruning_process}
\end{figure*}

\section{Probabilistic Robustness Analysis under Evaluator Perturbations}
\label{appsec:flip_prob}

We provide a probabilistic justification for the robustness behavior observed in RQ3 under evaluator perturbations. Let $\mathcal{M}$ denote the set of valid candidate graphs generated during optimization, where each $G \in \mathcal{M}$ satisfies the true constraint $S(G)\ge S(G_0)$, and let $\mathcal{U}$ denote the set of invalid candidates with $S(G)<S(G_0)$.
Define the top-tier valid set as
\[
\mathcal{M}^\ast=\{v\in\mathcal{M}\mid T(v)<\min_{u\in\mathcal{U}}T(u)\},
\]
which contains valid candidates whose makespan is smaller than any invalid candidate. The optimization procedure outputs
\[
G^\ast=\arg\min_{G\in\mathcal{W}}T(G),
\]
where
\[
\mathcal{W}=\{G\in\mathcal{M}\cup\mathcal{U}\mid S_\theta(G)\ge S_\theta(G_0)\}
\]
is the retained set under the Graph Evaluator.
Assume an otherwise accurate evaluator whose keep/reject decision at Algorithm~\ref{alg:lawa_overview} (line 12) is independently perturbed (flipped) with probability $\alpha$.
Then two events must simultaneously happen for an invalid final output: (i) at least one invalid candidate is falsely kept, and (ii) all top-tier valid candidates are falsely rejected.
Conditioned on $(\mathcal{U},\mathcal{M}^\ast)$, we have
\begin{align}
\Pr(G^\ast\in\mathcal{U}\mid\mathcal{U},\mathcal{M}^\ast)
&\le \Pr(\exists u\in\mathcal{U},\,u\in\mathcal{W})
\cdot
\Pr(\forall v\in\mathcal{M}^\ast,\,v\notin\mathcal{W}) \notag\\
&=\Big(1-(1-\alpha)^{|\mathcal{U}|}\Big)\alpha^{|\mathcal{M}^\ast|} \notag\\
&\le \Big(1-(1-\alpha)^{\bar N}\Big)\alpha^{|\mathcal{M}^\ast|},
\end{align}
where $\bar N$ is an upper bound on the number of candidates.

Let $p_0=\Pr(|\mathcal{M}^\ast|=0)$.
Marginalizing over candidate sets yields
\begin{align}
\Pr(G^\ast\in\mathcal{U})
&=\sum_{\mathcal{U},\mathcal{M}^\ast}
\Pr(G^\ast\in\mathcal{U}\mid\mathcal{U},\mathcal{M}^\ast)\Pr(\mathcal{U},\mathcal{M}^\ast) \notag\\
&\le
\Big(1-(1-\alpha)^{\bar N}\Big)\sum_{\mathcal{M}^\ast}\alpha^{|\mathcal{M}^\ast|}\Pr(\mathcal{M}^\ast) \notag\\
&=
\Big(1-(1-\alpha)^{\bar N}\Big)
\left(\sum_{k=0}^{\infty}\Pr(|\mathcal{M}^\ast|=k)\alpha^k\right) \notag\\
&\le
\Big(1-(1-\alpha)^{\bar N}\Big)\big(p_0+(1-p_0)\alpha\big).
\end{align}

Therefore, the probability of an invalid final output is explicitly determined by $\alpha$.
As $\alpha\to 0$, the upper bound also approaches $0$, consistent with RQ3, where performance remains stable under mild perturbations.

\begin{table}[t]
\centering
\caption{\small \textbf{Behavior of LAWA under limited executors.} Here $m$ is the number of available executors; runtime improvement is the relative reduction versus the initial workflow; operation counts are total rewiring/node-removal edits aggregated over all generated instances (7 benchmarks $\times$ 3 runs).}
\label{tab:appendix_limited_executors}
\resizebox{0.6\columnwidth}{!}{%
\begin{tabular}{ccc}
\toprule
\rowcolor{rowgray}
\textbf{$m$} & \textbf{Runtime Improvement (\%)} & \textbf{Operation Counts (rewire/remove)} \\
\midrule
1 & 43.67 & 19/64 \\
\rowcolor{rowgray}
2 & 65.38 & 83/61 \\
3 & 67.74 & 91/59 \\
\rowcolor{rowgray}
4 & 67.89 & 92/59 \\
\bottomrule
\end{tabular}}
\end{table}

\begin{table}[t]
\centering
\caption{\small \textbf{Accuracy robustness under rewiring-score approximation on three benchmarks.} We compare LAWA with and without first-order approximation when filtering rewiring edits. Values are accuracy (\%, mean$\stdsub{\text{std}}$) over 3 runs. Results show that the approximation does not introduce a noticeable drop in accuracy.}
\label{tab:appendix_approximation}
\resizebox{0.5\columnwidth}{!}{%
\begin{tabular}{lccc}
\toprule
\rowcolor{rowgray}
\textbf{Setting} & \textbf{MATH500} & \textbf{HumanEval} & \textbf{AQuA} \\
\midrule
w/o approximation & $67.00\stdsub{0.71}$ & $93.29\stdsub{0.00}$ & $83.86\stdsub{0.56}$ \\
\rowcolor{rowgray}
w/ approximation & $67.00\stdsub{0.86}$ & $92.89\stdsub{1.04}$ & $83.46\stdsub{0.85}$ \\
\bottomrule
\end{tabular}}
\end{table}

\section{More Experiment Results}
\label{appsec:add_ablations}

\subsection{Token Consumption}
\label{appsec:rq1_token}

To examine whether parallel rewiring increases context-broadcasting overhead, we report total token consumption (input + output) of the generated workflows for all methods under the same final setting.
As shown in Table~\ref{tab:appendix_token}, LAWA consistently remains token-efficient. The token consumption is substantially lower than heavily parallel baselines such as DyLAN.
This is because total token usage is largely determined by the number of agent invocations, and LAWA removes redundant nodes during optimization.

\begin{table*}[t]
\centering
\caption{\small \textbf{Total token consumption per query (Tokens$\downarrow$) on seven benchmarks.} We report total token consumption (input + output) for all methods. Values are mean$\stdsub{\text{std}}$ over three runs, and the best value in each benchmark is bolded.}
\label{tab:appendix_token}
\resizebox{0.99\textwidth}{!}{%
\begin{tabular}{lccccccc}
\toprule
\rowcolor{rowgray}
\textbf{Method}
& \textbf{MBPP}
& \textbf{HumanEval}
& \textbf{AQuA}
& \textbf{GSM8K}
& \textbf{MATH500}
& \textbf{SVAMP}
& \textbf{MMLU} \\
\midrule
Rand. Graph
& $22768\stdsub{1949}$
& $34403\stdsub{4326}$
& $64483\stdsub{2350}$
& $40409\stdsub{8685}$
& $57252\stdsub{8018}$
& $40825\stdsub{1684}$
& $48050\stdsub{6204}$ \\

\rowcolor{rowgray}
Rand. Chain
& $40289\stdsub{4229}$
& $49315\stdsub{3638}$
& $99064\stdsub{11753}$
& $83317\stdsub{446}$
& $103547\stdsub{5895}$
& $71156\stdsub{1235}$
& $73758\stdsub{821}$ \\

3-Round Disc.
& $219103\stdsub{7310}$
& $243745\stdsub{3320}$
& $595536\stdsub{12769}$
& $500286\stdsub{2419}$
& $646475\stdsub{21625}$
& $445061\stdsub{7785}$
& $470002\stdsub{1510}$ \\

\midrule
\rowcolor{rowgray}
DyLAN
& $55388\stdsub{3921}$
& $72015\stdsub{5500}$
& $85396\stdsub{9330}$
& $41915\stdsub{20638}$
& $98495\stdsub{3777}$
& $62078\stdsub{2553}$
& $81276\stdsub{1301}$ \\

G-Designer
& $37128\stdsub{2550}$
& $44954\stdsub{1157}$
& $88371\stdsub{4124}$
& $80719\stdsub{3347}$
& $99698\stdsub{4556}$
& $70241\stdsub{5089}$
& $67340\stdsub{2960}$ \\

\rowcolor{rowgray}
AgentDropout
& $19870\stdsub{9307}$
& $27288\stdsub{10278}$
& $54616\stdsub{5328}$
& $29668\stdsub{7334}$
& $46292\stdsub{7966}$
& $41654\stdsub{4011}$
& $46553\stdsub{10920}$ \\

AgentPrune
& $34834\stdsub{5788}$
& $46360\stdsub{222}$
& $88365\stdsub{3343}$
& $75535\stdsub{1079}$
& $99848\stdsub{3995}$
& $59406\stdsub{5280}$
& $62790\stdsub{5324}$ \\

\rowcolor{rowgray}
BiRouter
& $11338\stdsub{2789}$
& $16150\stdsub{3886}$
& $14843\stdsub{3639}$
& $\textbf{9344}\stdsub{929}$
& $15100\stdsub{3514}$
& $8448\stdsub{227}$
& $\textbf{14763}\stdsub{3666}$ \\
AFlow
& $11865\stdsub{3037}$
& $16532\stdsub{3825}$
& $35135\stdsub{4172}$
& $23594\stdsub{2120}$
& $35469\stdsub{4828}$
& $22166\stdsub{2675}$
& $24833\stdsub{4578}$ \\
\rowcolor{rowgray}
LAMaS
& $7898\stdsub{150}$
& $11705\stdsub{240}$
& $14047\stdsub{3976}$
& $15841\stdsub{6533}$
& $17849\stdsub{2544}$
& $13766\stdsub{2070}$
& $15012\stdsub{103}$ \\

\midrule
\textbf{LAWA}
& $\textbf{6516}\stdsub{30}$
& $\textbf{11401}\stdsub{4099}$
& $\textbf{10796}\stdsub{4121}$
& $10461\stdsub{1862}$
& $\textbf{12207}\stdsub{3152}$
& $\textbf{7615}\stdsub{2304}$
& $15077\stdsub{2897}$ \\
\bottomrule
\end{tabular}}
\end{table*}

\subsection{Optimization Runtime}
\label{appsec:rq1_overhead}

To better characterize computational overhead of LAWA, we report both (i) wall-clock time to generate one workflow and (ii) Graph Evaluator training time.
Unless otherwise stated, runtime measurements in this section are conducted on a machine with an AMD EPYC 9655 96-Core Processor CPU.
As shown in Table~\ref{tab:appendix_gen_time}, workflow generation typically takes only around 1--2.3 seconds because the search operates only on graph structure without executing real agent calls. Compared with high-latency baseline workflows in Table~\ref{tab:time_results}, this search overhead is small relative to recurrent execution costs.

As shown in Table~\ref{tab:appendix_train_time}, Graph Evaluator training takes 21--95 seconds per benchmark and is performed once per benchmark as an offline preprocessing step. Therefore, when LAWA is used repeatedly to replace high-latency workflows, this one-time overhead can be easily amortized over subsequent executions. For example, using Random Chain as a representative high-latency initial workflow, the one-time evaluator-training and search cost is offset after roughly 2--7 recurrent executions across benchmarks.

\paragraph{Cost of graph-dataset construction.}
On MATH500, constructing the full graph dataset (evaluating 2048 randomly generated workflows on 50 training samples) takes 38.7 minutes of wall-clock time on 8 NVIDIA RTX PRO 6000 GPUs. This is a one-time, fully parallelizable offline stage: unlike on-policy baselines, whose data collection is sequential and repeated during optimization, LAWA's workflow evaluations have no sequential dependency. On the same server, the complete LAWA pipeline (dataset generation, training, and workflow generation) takes 41.2 minutes, whereas AgentDropout requires 129.6 minutes under the same 2048-rollout budget (64 iterations $\times$ batch size 32). The gap arises because evaluating 2048 workflows in one pass fully utilizes the inference server: it costs only about $20\times$ the time of evaluating 32 workflows, despite processing $64\times$ as many workflows. Thus, even after including graph-dataset construction, LAWA takes about one third of the wall-clock time of this on-policy baseline.

\begin{table}[t]
\centering
\caption{\small \textbf{Wall-clock time to generate a single workflow (Seconds$\downarrow$).}}
\label{tab:appendix_gen_time}
\resizebox{0.25\columnwidth}{!}{%
\begin{tabular}{lc}
\toprule
\rowcolor{rowgray}
\textbf{Benchmark} & \textbf{Time (s)} \\
\midrule
MBPP & $2.2736\stdsub{0.2287}$ \\
\rowcolor{rowgray}
HumanEval & $1.1563\stdsub{0.4062}$ \\
AQuA & $1.6917\stdsub{0.7676}$ \\
\rowcolor{rowgray}
GSM8K & $2.2471\stdsub{0.1613}$ \\
MATH500 & $1.7981\stdsub{0.1430}$ \\
\rowcolor{rowgray}
SVAMP & $2.2039\stdsub{0.2768}$ \\
MMLU & $2.1496\stdsub{0.1832}$ \\
\bottomrule
\end{tabular}}
\end{table}

\begin{table}[t]
\centering
\caption{\small \textbf{Wall-clock time to train the Graph Evaluator (Seconds$\downarrow$).}}
\label{tab:appendix_train_time}
\resizebox{0.25\columnwidth}{!}{%
\begin{tabular}{lc}
\toprule
\rowcolor{rowgray}
\textbf{Benchmark} & \textbf{Train Time (s)} \\
\midrule
MBPP & 44.45 \\
\rowcolor{rowgray}
HumanEval & 21.32 \\
AQuA & 63.34 \\
\rowcolor{rowgray}
GSM8K & 34.65 \\
MATH500 & 46.47 \\
\rowcolor{rowgray}
SVAMP & 23.32 \\
MMLU & 95.08 \\
\bottomrule
\end{tabular}}
\end{table}

\subsection{Frontier Backbone Model}
\label{appsec:frontier}

We additionally evaluate the officially hosted frontier model Kimi-K2.6 on MBPP, MATH500, and MMLU, comparing LAWA with DyLAN and AgentDropout. For LAWA, we reuse the original Graph Evaluator, Graph Generator, and estimated $t_k$ values, without retraining the learned modules or recalibrating the latency estimates. As shown in Table~\ref{tab:appendix_kimi}, LAWA's performance advantage is preserved with the stronger backbone: it remains best on MBPP and MMLU, and is on par with DyLAN on MATH500. A stronger backbone makes each individual agent more capable, so a compact workflow can attain high accuracy with fewer agents, which favors LAWA's ability to remove redundant agents. LAWA also remains the fastest method on all three benchmarks.

\begin{table}[t]
\centering
\caption{\small \textbf{Accuracy (\%$\uparrow$) and real wall-clock makespan (s$\downarrow$) with Kimi-K2.6.} Values are mean$\stdsub{\text{std}}$ over three workflow graphs.}
\label{tab:appendix_kimi}
\resizebox{0.9\columnwidth}{!}{%
\begin{tabular}{lcccccc}
\toprule
\rowcolor{rowgray}
& \multicolumn{2}{c}{\textbf{MBPP}} & \multicolumn{2}{c}{\textbf{MATH500}} & \multicolumn{2}{c}{\textbf{MMLU}} \\
\cmidrule(lr){2-3} \cmidrule(lr){4-5} \cmidrule(lr){6-7}
\rowcolor{rowgray}
\textbf{Method} & \textbf{Acc.} & \textbf{Time} & \textbf{Acc.} & \textbf{Time} & \textbf{Acc.} & \textbf{Time} \\
\midrule
DyLAN & $73.41\stdsub{0.22}$ & $40.27\stdsub{5.92}$ & $81.20\stdsub{1.25}$ & $32.01\stdsub{9.53}$ & $80.73\stdsub{1.80}$ & $33.07\stdsub{1.03}$ \\
\rowcolor{rowgray}
AgentDropout & $72.24\stdsub{0.59}$ & $43.66\stdsub{2.37}$ & $79.27\stdsub{2.89}$ & $43.34\stdsub{1.98}$ & $79.82\stdsub{3.61}$ & $52.66\stdsub{13.15}$ \\
\midrule
\textbf{LAWA} & $\textbf{73.80}\stdsub{0.81}$ & $\textbf{14.19}\stdsub{1.21}$ & $\textbf{81.27}\stdsub{1.22}$ & $\textbf{18.43}\stdsub{2.99}$ & $\textbf{81.74}\stdsub{1.77}$ & $\textbf{19.75}\stdsub{2.10}$ \\
\bottomrule
\end{tabular}}
\end{table}

\subsection{Finite-Resource Scheduling}
\label{appsec:finite_resource}
We study how LAWA's editing policy changes when the available parallel resources are constrained. 
During graph editing, we replace the default critical-path estimator (\cref{eq:def_T}) with an HLFET list-scheduling runtime estimator~\cite{hlfet} and vary the number of executors $m\in\{1,2,3,4\}$. Runtime improvement is measured as the relative reduction against the initial workflow, and operation counts are aggregated over all generated instances (7 benchmarks $\times$ 3 runs). 
As reported in Table~\ref{tab:appendix_limited_executors}, the optimization behavior shows a clear transition: under single-executor constraints, LAWA mainly shortens the critical chain through node removal, while with more executors it increasingly prefers rewiring to unlock parallel branches; correspondingly, runtime gains improve quickly at first and then gradually plateau as available parallelism approaches saturation. 
This indicates that, even under limited executors, LAWA can adapt its editing strategy to the available resources.

\subsection{Robustness under First-Order Approximation}
We examine whether the first-order gradient approximation used to filter rewiring operations harms final optimization quality. 
We compare default LAWA (with approximation) against a variant that disables approximation and computes rewiring-score changes exactly during candidate generation, then evaluate final accuracy on MATH500, HumanEval, and AQuA.
As shown in Table~\ref{tab:appendix_approximation}, the performance differences are minor across benchmarks, indicating that the approximation does not introduce a noticeable drop in final accuracy while helping reduce the number of candidate graphs.

\subsection{Effectiveness of the Graph Evaluator}
We conduct an ablation study to examine whether the Graph Evaluator stabilizes LAWA's iterative graph editing process. Starting from the same randomly initialized chain workflows, we compare the default LAWA against a variant without evaluator-based filtering, which keeps candidates only according to makespan. Table~\ref{tab:appendix_eval_effectiveness} reports the downstream-performance change across optimization steps on two representative benchmarks. With the Graph Evaluator, LAWA yields positive mean performance changes across optimization steps, whereas removing the evaluator leads to less stable performance and can result in degradation. These results suggest that the Graph Evaluator helps retain downstream task quality while the structural optimizer reduces latency.

\begin{table*}[t]
\centering
\caption{\small \textbf{Performance change across optimization steps with and without the Graph Evaluator.} Values are downstream-performance changes (\%, mean$\stdsub{\text{std}}$) relative to the initial workflow, computed over eight different random initial graphs.}
\label{tab:appendix_eval_effectiveness}
\resizebox{\textwidth}{!}{%
\begin{tabular}{llccccc}
\toprule
\rowcolor{rowgray}
\textbf{Dataset} & \textbf{Method} & \textbf{Step 1} & \textbf{Step 2} & \textbf{Step 3} & \textbf{Step 4} & \textbf{Step 5} \\
\midrule
HumanEval & \textbf{LAWA (w/ Evaluator)} & $\textbf{+1.68}\stdsub{1.99}$ & $\textbf{+2.21}\stdsub{2.97}$ & $\textbf{+2.29}\stdsub{3.22}$ & $\textbf{+2.44}\stdsub{3.37}$ & $\textbf{+1.75}\stdsub{2.54}$ \\
\rowcolor{rowgray}
HumanEval & w/o Evaluator & $-0.30\stdsub{2.14}$ & $-0.15\stdsub{2.82}$ & $-0.07\stdsub{2.26}$ & $-0.53\stdsub{2.12}$ & $-2.74\stdsub{1.72}$ \\
MATH500 & \textbf{LAWA (w/ Evaluator)} & $\textbf{+0.95}\stdsub{0.60}$ & $\textbf{+1.65}\stdsub{1.09}$ & $\textbf{+1.25}\stdsub{1.17}$ & $\textbf{+1.65}\stdsub{1.14}$ & $\textbf{+1.35}\stdsub{1.23}$ \\
\rowcolor{rowgray}
MATH500 & w/o Evaluator & $-0.20\stdsub{0.63}$ & $-0.35\stdsub{1.19}$ & $-0.45\stdsub{1.19}$ & $-0.63\stdsub{1.34}$ & $-1.72\stdsub{1.51}$ \\
\bottomrule
\end{tabular}}
\end{table*}

\subsection{Graph Evaluator Architecture Sensitivity}
We evaluate how sensitive the Graph Evaluator is to GNN architecture choices and justify the default backbone used in LAWA. 
On HumanEval, we run a full ablation over GNN type (\texttt{gcn}, \texttt{gat}, \texttt{transformer}), depth (1--4 layers), and hidden size (256/512), while keeping all other training settings fixed and reporting train/test accuracy over 3 seeds. 
Table~\ref{tab:appendix_gnn_ablation} indicates that GCN performs most consistently across configurations, while the performance of GAT and Transformer varies more with depth and hidden size. The 3-layer Transformer with hidden dimension 512 achieves both train and test accuracy above 95\%, and Graph Transformers are generally more suited to capture the long-range interactions that may arise from multi-hop dependencies in workflow graphs~\cite{ying2021do}. Therefore, we use a 3-layer Transformer with hidden dimension 512 as the default Graph Evaluator.

\section{Limitations and Future Works}
\label{appsec:limitations}

This work considers a static agent pool and a fixed set of downstream tasks, a standard setting that is widely adopted in prior work. In practice, agent availability and task distributions may evolve over time. Extending the framework to more dynamic scenarios is a promising direction for future work. 

In dynamic environments, the set of available agents may change over time, leading to changes in the feasible workflows and the associated makespan. Extending the framework to accommodate a dynamic agent pool would require the workflow to adapt to such changes while maintaining efficiency. Similarly, downstream tasks may evolve over time, potentially altering the desirable workflow patterns. Developing mechanisms that allow the evaluator and workflow to adapt to such shifts in a data-efficient and incremental manner, rather than being rebuilt from scratch, is another important direction for future work.

%
%

\begin{table*}[t]
\centering
\caption{\small \textbf{GNN architecture ablation on HumanEval for Graph Evaluator.} We vary GNN type (\texttt{gcn}, \texttt{gat}, \texttt{transformer}), number of layers (\{1,2,3,4\}), and hidden dimension (\{256,512\}), while keeping other training settings fixed. Values are train/test accuracy (\%, mean$\stdsub{\text{std}}$) over 3 seeds. Rows where both train and test mean exceed 95 are bolded.}
\label{tab:appendix_gnn_ablation}
\vspace{5mm}
\resizebox{0.7\textwidth}{!}{%
\begin{tabular}{ccccc}
\toprule
\rowcolor{rowgray}
\textbf{Hidden Dim} & \textbf{GNN Type} & \textbf{\#Layers} & \textbf{Train Acc (\%)} & \textbf{Test Acc (\%)} \\
\midrule
256 & gcn & 1 & $94.68\stdsub{0.68}$ & $94.79\stdsub{3.35}$ \\
\rowcolor{rowgray}
256 & gcn & 2 & $94.63\stdsub{0.49}$ & $95.17\stdsub{4.28}$ \\
256 & gcn & 3 & $\textbf{95.05}\stdsub{0.44}$ & $\textbf{95.14}\stdsub{2.58}$ \\
\rowcolor{rowgray}
256 & gcn & 4 & $94.93\stdsub{0.20}$ & $95.77\stdsub{3.30}$ \\
256 & gat & 1 & $94.97\stdsub{0.73}$ & $94.47\stdsub{3.27}$ \\
\rowcolor{rowgray}
256 & gat & 2 & $95.38\stdsub{1.27}$ & $94.00\stdsub{3.74}$ \\
256 & gat & 3 & $94.22\stdsub{1.63}$ & $93.57\stdsub{3.74}$ \\
\rowcolor{rowgray}
256 & gat & 4 & $95.00\stdsub{0.77}$ & $92.44\stdsub{3.87}$ \\
256 & transformer & 1 & $94.18\stdsub{0.58}$ & $93.47\stdsub{4.36}$ \\
\rowcolor{rowgray}
256 & transformer & 2 & $95.35\stdsub{1.47}$ & $94.43\stdsub{4.35}$ \\
256 & transformer & 3 & $94.75\stdsub{1.70}$ & $94.03\stdsub{4.14}$ \\
\rowcolor{rowgray}
256 & transformer & 4 & $93.60\stdsub{1.58}$ & $94.00\stdsub{2.90}$ \\
512 & gcn & 1 & $94.92\stdsub{1.39}$ & $95.11\stdsub{3.80}$ \\
\rowcolor{rowgray}
512 & gcn & 2 & $\textbf{95.43}\stdsub{0.68}$ & $\textbf{95.45}\stdsub{2.94}$ \\
512 & gcn & 3 & $95.07\stdsub{0.80}$ & $94.17\stdsub{4.03}$ \\
\rowcolor{rowgray}
512 & gcn & 4 & $94.85\stdsub{0.58}$ & $95.28\stdsub{1.95}$ \\
512 & gat & 1 & $95.00\stdsub{1.54}$ & $94.90\stdsub{2.21}$ \\
\rowcolor{rowgray}
512 & gat & 2 & $95.13\stdsub{1.16}$ & $93.40\stdsub{4.35}$ \\
512 & gat & 3 & $95.13\stdsub{1.09}$ & $93.82\stdsub{2.16}$ \\
\rowcolor{rowgray}
512 & gat & 4 & $95.70\stdsub{0.82}$ & $93.25\stdsub{3.89}$ \\
512 & transformer & 1 & $93.12\stdsub{1.45}$ & $91.67\stdsub{2.64}$ \\
\rowcolor{rowgray}
512 & transformer & 2 & $93.23\stdsub{0.35}$ & $89.93\stdsub{6.86}$ \\
512 & transformer & 3 & $\textbf{95.97}\stdsub{1.29}$ & $\textbf{95.11}\stdsub{4.39}$ \\
\rowcolor{rowgray}
512 & transformer & 4 & $95.28\stdsub{3.68}$ & $94.68\stdsub{4.06}$ \\
\bottomrule
\end{tabular}}
\end{table*}

\clearpage



\end{document}